\documentclass[11pt]{article}
\usepackage{authblk}
\usepackage{amssymb}
\usepackage{amsthm}
\usepackage{mathtools}
\usepackage{amsfonts}
\usepackage{amsmath}
\usepackage{graphicx} 
\usepackage{amsmath,amsthm,amssymb}
\usepackage[margin=1in]{geometry}
\usepackage[colorlinks=true]{hyperref}
\usepackage{thmtools, thm-restate}
\usepackage{algorithm}
\usepackage{algorithm}
\usepackage{algpseudocode}

\declaretheorem{theorem}
\usepackage{float}

\newcommand{\PROOF}{\begin{proof}}

\newcommand{\QED}{\end{proof}}

\newcommand{\range}{\mathrm{range}}

\newcommand{\prefix}{\sqsubseteq}

\newcommand{\restr}{\upharpoonright}

\newcommand{\N}{\mathbb{N}}

\newcommand{\Q}{\mathbb{Q}}

\newcommand{\poly}{{\mathrm{poly}}}

\newcommand{\cdim}{\mathrm{cdim}}

\newcommand{\cDim}{{\mathrm{cDim}}}

\newcommand{\D}{{\mathcal{D}}}

\newcommand{\NP}{{\ensuremath{\mathrm{NP}}}}

\newcommand{\PSPACE}{{\rm PSPACE}}
\renewcommand{\P}{\ensuremath{{\mathrm P}}}

\newcommand{\F}{{\mathcal{F}}}

\newcommand{\supp}{\mathrm{supp}}

\newcommand{\Rplus}{[0,\infty)}
\renewcommand{\F}{\mathcal{F}}

\usepackage[inline]{enumitem}

\newtheorem{lemma}{Lemma} 
 
\newtheorem{corollary}{Corollary}
\newtheorem{definition}{Definition}
\newtheorem{question}{Question}

\newtheorem*{definition*}{Definition}

\newcommand{\nifty}{n \rightarrow \infty}

\renewcommand{\S}{{X}} 

\newcommand{\Sstr}{S_{\mathrm{str}}}

\newcommand{\cdimp}{{\cdim}_\mathrm{P}}
	
\newcommand{\cDimp}{{\cDim}_\mathrm{P}}
\newcommand{\pr}{\mathrm{Pr}}

\newcommand{\Kpoly}{\mathcal{K}_{\poly}}
\newcommand{\Kstrpoly}{\mathcal{K}^{\mathrm{str}}_{\poly}}

\newcommand{\constc}{\mathbf{c}}

\title{One-Way Functions and Polynomial Time Dimension}

\author[1]{Satyadev Nandakumar}
\author[2]{Subin Pulari}
\author[1]{Akhil S}
\author[1]{Suronjona Sarma}
\affil[1]{
  Department of Computer Science and Engineering\\
  Indian Institute of Technology Kanpur,
  Kanpur, Uttar Pradesh, India.
}
\affil[1]{\{satyadev,akhis,suronjona\}@cse.iitk.ac.in}
\affil[2]{
Universit\'e de Bordeaux, CNRS, Bordeaux INP, LaBRI, UMR 5800, F-33400, Talence, France
}
\affil[2]{subin.pulari@labri.fr}

\date { }

\begin{document}

\maketitle

\begin{abstract}
	
	 This paper demonstrates a duality between the non-robustness of polynomial time dimension and the existence of one-way functions. Polynomial-time dimension (denoted $\mathrm{cdim}_\mathrm{P}$) quantifies the density of information of infinite sequences using polynomial time betting algorithms called $s$-gales. An alternate quantification of the notion of polynomial time density of information is using polynomial-time Kolmogorov complexity rate (denoted
	$\mathcal{K}_\text{poly}$). Hitchcock and Vinodchandran (CCC 2004) showed that $\mathrm{cdim}_\mathrm{P}$ is always greater than or equal to $\mathcal{K}_\text{poly}$. We first show that if one-way functions exist then there exists a polynomial-time samplable distribution with respect to which $\mathrm{cdim}_\mathrm{P}$ and $\mathcal{K}_\text{poly}$ are separated by a uniform gap with  probability $1$. Conversely, we show that if there exists such a polynomial-time samplable distribution, then (infinitely-often) one-way functions exist. 	  
	    
  Using our main results, we solve a long standing open problem posed by Hitchcock and Vinodchandran (CCC 2004) and Stull under the assumption that one-way functions exist. We demonstrate that if one-way functions exist, then there are individual sequences $X$ whose poly-time dimension strictly exceeds  $\mathcal{K}_\text{poly}(X)$, that is $\mathrm{cdim}_\mathrm{P}(X) > \mathcal{K}_\text{poly}(X)$. The corresponding unbounded notions,
  namely, the constructive dimension and the asymptotic lower rate of
  unbounded Kolmogorov complexity are equal for every sequence.
  Analogous notions are equal even at polynomial space and finite-state level. In view of these results, it is reasonable
  to conjecture that the polynomial-time quantities are identical for
  every sequence and set of sequences. However, under a plausible assumption which underlies modern
  cryptography - namely the existence of one-way functions, we \emph{refute}
  the conjecture thereby giving a negative answer to the open question posed by Hitchcock, Vinodchandran and Stull. Further, we show that the gap between
 these quantities can be made as large as possible (\emph{i.e.} close
 to 1). We also establish similar bounds for
 strong poly-time dimension versus asymptotic upper Kolmogorov
 complexity rates. 
 Our proof uses several new constructions and arguments involving probabilistic tools such as the
 Borel-Cantelli Lemma, the Kolmogorov inequality for martingales and  the theorem on universal extrapolation by Ilango, Ren, and Santhanam. 
 This work shows that the question of \emph{non-robustness} of
 polynomial-time information density notions, which is \emph{prima
 	facie} different, is intimately related to questions which are of
 current interest in cryptography and \emph{meta-complexity}.
\end{abstract}
\newpage
\tableofcontents

\section{Introduction}

\subsection{Context and Motivation}

\textbf{One-way functions}
\cite{DiffieHellman1976,GoldreichCryptography,Levin2003} are functions
on finite strings that are easy to compute but are hard to invert,
except possibly on a negligible fraction of the input strings of a
given length. The concept of a one-way function is central in
cryptography, since the existence of such functions are both necessary
and sufficient for the existence of essential cryptographic primitives
like pseudorandom generators \cite{HILL99}, digital signatures
\cite{Rompel1990}, private key encryption \cite{HILL99} \cite{GM84},
authentication schemes \cite{FS90} and commitment schemes
\cite{Naor1991}. The question of existence of One-way functions
based on average case hardness of NP-complete problems has been one that has been central in cryptography.

Recent advancements in meta-complexity reveal intriguing
connections between existence of one-way functions and the computational hardness
of problems studied in meta-complexity. Characterisations of existence of One-way functions have been studied based on average case hardness of time bounded
Kolmogorov complexity \cite{liu2020one}, average case hardness of $\mathrm{McKTP}$ \cite{allender2021one},  average-case hardness of approximating Kolmogorov complexity on samplable distributions  \cite{ilango2022robustness}, average case failure of symmetry of information \cite{hirahara2023duality}, average-case easiness of approximating pKt complexity \cite{HirhLuOliveria}, and NP-hardness of distributional Kolmogorov complexity under randomized polynomial-time reductions \cite{Hirahara23}.

In this paper, we show that the existence of one-way functions implies the existence of dimension gaps on polynomial time samplable distributions over infinite sequences. Furthermore, we show that existence of such distributions implies existence of infinitely often one way functions.
Using this result,
we give a surprising negative resolution to the longstanding open question of the robustness of polynomial-time dimension posed by Hitchcock, Vinodchandran, and Stull \cite{dercjournal, stullsurvey}.

Polynomial-time dimension quantifies the
asymptotic rate of information in
an infinite sequence of bits. There are several approaches towards defining polynomial time dimension. The first approach analyzes the asymptotic \emph{compressibility} of the finite prefixes of the given infinite sequence. The polynomial-time bounded \emph{Kolmogorov Complexity} of a string $x$ is the length of the shortest program that can output $x$ in at most polynomial number of time steps. In the \emph{compressibility} approach, the polynomial time density of information is defined in terms of the polynomial-time bounded Kolmogorov complexity of the prefixes of the string, denoted $\Kpoly$. In contrast, we can also use the \emph{gambling based} approach to quantify polynomial-time density of information. Polynomial-time dimension, denoted $\cdimp$, is defined using polynomial-time betting algorithms known as $s$-gales. These algorithms attempt to achieve profit by placing successive bets on the bits of the sequence. Each bet made by the $s$-gale on a symbol reflects its confidence in the occurrence of that corresponding symbol.

In this work, we initiate a line of investigation into the connection between cryptographic primitives and the
\emph{robustness} of complexity notions of infinite sequences. We consider polynomial time samplable distribution $\nu$ over $\Sigma^\infty$ that uses at most $sn$ random bits (where $s < 1$) to sample a string $w \in \Sigma^n$. We call such distributions as \emph{short-seed  polynomial time samplable distributions}.
We show that the existence of one-way functions implies
dimension gaps between	$\cdimp$ and $\Kpoly$ over a collection of sequences in $\Sigma^\infty$, such that the collection has probability $1$ according to a short seed polynomial time samplable distribution. Futhermore, we show that the existence of such distributions with dimension gaps implies the existence of infinitely-often one-way functions. 

Formally establishing the connection between the robustness of these information density notions and one-way functions in the \emph{infinite string} setting requires significant technical effort and new insights, extending beyond the well-known relationships between compressibility and indistinguishability \cite{HILL99} \cite{Yao82} \cite{Kabanets2000} in the setting of finite strings .

\subsection{Applications}

Considerable research has been devoted to understanding whether
fundamental properties of unbounded Kolmogorov complexity survive in
the time bounded setting. Recently, it was shown that the  symmetry of information of time-bounded Kolmogorov
complexity is equivalent to the existence of OWFs \cite{hirahara2023duality}.
Along similar lines, we investigate whether the characterization of
constructive dimension in terms of Kolmogorov complexity
\cite{LutzDimension2003,Mayordomo02} survives in the time
bounded setting.  This is the question of \emph{robustness of polynomial-time dimension}  was posed by Hitchcock and Vinodchandran \cite{dercconference, dercjournal}, and  later  mentioned in a survey by Stull (see Open Question 3.71 in \cite{stullsurvey}). As an application of our main result, we provide a surprising negative resolution to this longstanding open question, under the assumption that one-way functions exist.  

The unbounded analogue of time-bounded dimension is known to be robust with
markedly different approaches to its definition - via predictors, via
unbounded time $s$-gales, and via unbounded Kolmogorov complexity
rates - all known to be equivalent \cite{Mayordomo02} \cite{LutzDimension2003}. 
At the other extreme,
\emph{finite-state} analogues of these notions are also known to be
robust. Hitchcock and Vinodchandran \cite{dercconference} show that the
polynomial-dimension of every sequence is at least the asymptotic
lower density of the polynomial-time Kolmogorov complexity of its prefixes. Given the
robustness at other levels and the fact that the
inequality holds in one direction, it is natural to conjecture that
the inequality holds in the other direction as well. This would establish the
robustness of polynomial time dimension.
However, the question of
robustness of dimension at the important ``intermediate''
resource-bounded level of polynomial time dimension, has remained open for a
long time.


Using our main theorem, we establish that the existence of one-way functions implies
dimension gaps over a \emph{suitable collection} of sequences. Thus, we show that polynomial time dimension of sets is non-robust if one-way functions exist. We then extend this result to show the stronger result that polynomial time dimension of sequences is non-robust if one-way functions exist. 
Furthermore, we show that the gap between $\Kpoly$ and $\cdimp$ can be arbitrarily close to $1$. Additionally, we show that \emph{strong} polynomial
time-bounded dimension is also non-robust under the same assumption.


The study of \emph{meta-complexity} - the complexity of computing various measure in complexity theory has been shown to have close connections to cryptographic notions. Our work broadens the scope of this connection.  Our main results establish that the study of polynomial-time dimension and, more generally, information-theoretic complexity notions of \emph{infinite sequences} is closely connected with the existence of fundamental primitives in complexity and cryptography. In doing so, we also provide an unexpected resolution to a long-standing open problem.



\subsection{Technical Background}
\subsubsection*{s-gales and Kolmogorov Complexity}

\textbf{$s$-gales} are betting strategies on infinite sequences. The
betting game on an infinite binary sequence $X \in \Sigma^\infty$ can
be understood as follows. The player starts with an initial capital of
$d(\lambda)=1$ on the empty string. Here $d$ is the betting function
(capital) of the $s$-gale. At the $n$\textsuperscript{th} stage,
having accumulated a capital $d(x_1x_2\dots x_n)$ on the first $n$
bits of the sequence, the player is allowed to place an amount as a
``bet" on the next bit $x_{n+1}$ of the sequence. The rule of the game
is that the expected value of capital received by the player after the
bet is $2^s$ times the capital they started with, that is $2^s . d(w)
= d(w0) + d(w 1)$, for any finite string $w \in \Sigma^*$.  The player $d$ succeeds on
the sequence $X \in \Sigma^\infty$ if they can secure an arbitrary
amount of capital over the course of betting on the infinite sequence,
more precisely if $\limsup_{n \to \infty} d(x_1x_2\dots x_n) =
\infty$. 

When $s=1$, the average capital after the bet on $w0$ or $w1$ is equal
to the previous capital $d(w)$ (when $s=1$, such $s$-gales are
referred to as \emph{martingales}). Therefore, the betting strategy is
fair. When $s<1$, the average capital after the bet is strictly less
than the previous capital $d(w)$. Therefore, such betting strategies
are inherently unfair to the player. When $s=1$, the player can bet ``evenly" on the next bit if they are unsure.
But as $s$ decreases, the setting becomes more unfavorable and the
player needs to bet more aggressively to keep increasing their capital. 

The \textbf{Kolmogorov Complexity} of a finite string $x$, denoted by $K(x)$, is the length of the shortest program that produces the string as the output. 
Any string $x$ can be produced by a program that trivially outputs the string $x$, and therefore the Kolmogorov complexity of a string $x$ is less than or equal to its length, up to an additive constant. However, if the string has lesser amount of information, there may be programs of shorter length that outputs it. For instance, if all the even bits of the string are $0$, the program needs to encode only the bits at the odd indices to produce the string. Therefore, its Kolmogorov complexity is at most half its length, up to additive constants.

\subsubsection*{Polynomial-time Dimension}

\textbf{Resource bounded dimension} is defined by placing resource bounds on the computation of the $s$-gale. Analogously,  resource bounded Kolmogorov complexity is defined by placing resource bounds on the programs that can print a string.


The \textbf{Polynomial time dimension} quantifies the rate of information in an infinite sequence, measured with respect to polynomial time bounded computation. It is formulated using polynomial time $s$-gales. An $s$-gale $d$ wins on an infinite sequence $X$ if  $\limsup_{n \to \infty} d(X \restr n) = \infty$. $d$ is said to be polynomial-time computable if for some $p(n) \in \poly(n)$, $d$ takes at most $p(n)$ time to compute $d(w)$, where $n$ is the length of $w$.

\begin{definition*}[Polynomial-time dimension \cite{LutzDimension2003}]
	The polynomial-time dimension of $ \F \subseteq \Sigma^\infty$ is defined as 
	\begin{align*}
		\cdimp(\F) = \inf \{s \mid \exists \text{ a polynomial-time  } s\text{-gale } d \text{ such that d succeeds on all } X \in \F \}.
	\end{align*}
	For a sequence $X \in \Sigma^\infty$, define $ \cdimp(X) = \cdimp(\{X\})$.
\end{definition*}

That is $\cdimp(X)$ is the ``lowest" $s$ for which there exists a polynomial time $s$-gale $d$ that succeeds on $X$. 

\subsubsection*{Time bounded Kolmogorov Complexity}

On the other hand,  we can use \emph{time bounded} Kolmogorov complexity to define a polynomial time analogue of $\mathcal{K}$. 
 For a given time bound $t(n)$, the \textbf{$t$-time bounded Kolmogorov complexity} $K_t(x)$ of a string $x$ is the length of the shortest program that generates $x$ in at most $t(\lvert x \rvert)$ time steps (see \cite{LiVitanyi}). 

For a given time bound $t(n)$, the quantity $\liminf_{\nifty} K_t(X \restr n) / n$ gives the $t$-time bounded rate of information in an infinite string $X$. The polynomial time analogue of $\mathcal{K}$ (Definition \ref{def:Kdefinition}) is the infimum of this quantity over all polynomial time bounds $t$.

\begin{definition*}[\cite{dercconference,dercjournal}]
	\label{def:Kdefinition}
	For any $\F \subseteq \Sigma^\infty$,
	\begin{align*}
		\Kpoly(\F) = \inf\limits_{t \in \poly}\sup\limits_{X \in \F}\liminf\limits_{n \rightarrow \infty} \frac{K_t(X \restr n)}{n}.
	\end{align*}
\end{definition*}

\subsection*{Robustness of  Dimension}

\subsubsection*{Robustness at other levels}

In the classical setting, we use the notion of $\mathcal{K}$ defined using time-unbounded Kolmogorov complexity $K(x)$ and $\cdim$ defined using lower semi-computable $s$-gales.
Mayordomo \cite{Mayordomo02} and Lutz \cite{diss} show that the unbounded notions of $\mathcal{K}$ and $\cdim$ are equivalent, hence establishing the robustness of \emph{constructive} dimension.

\begin{theorem} [Mayordomo \cite{Mayordomo02} , Lutz
	\cite{diss}] \label{thm:cdimAndliminf} 
	For all $\F \subseteq \Sigma^\infty$, 
	$\cdim(\F) =
	\sup\limits_{X \in \F}\liminf\limits_{n \rightarrow \infty} \frac{K(X \restr n)}{n}.$ 
\end{theorem}

$\PSPACE$ analogues of $\cdim$ and $\mathcal{K}$ are defined by restricting the $s$-gales in the definition of $\cdim$ to be polynomial space computable and by using polynomial space bounded Kolmogorov complexity in the definition of $\mathcal{K}$. Hitchcock and Vinodchandran \cite{dercconference,dercjournal} showed that these notions of resource bounded dimension coincide in the $\PSPACE$ setting,  hence establishing the robustness of \emph{polynomial-space} dimension.

\begin{theorem}[\cite{dercconference,dercjournal}]
	For every $\F \subseteq \Sigma^\infty$, $\cdim_{\PSPACE}(\F) = \mathcal{K}_{\PSPACE}(\F)$. 
	
\end{theorem}

\subsubsection*{Robustness of Polynomial time Dimension}

 This leads to the following question: \emph{Are the notions of polynomial time dimension formulated using $s$-gales ($\cdimp$) and time bounded Kolmogorov complexity ($\Kpoly$) equivalent?}. 

Hitchcock and Vinodchandran \cite{dercjournal} showed that $\cdimp$ is always greater than or equal to $\Kpoly$\footnote{See Section \ref{sec:technicallemmas} for alternate proofs of Theorems \ref{thm:hitchcockvinodchandran} and \ref{thm:hitchcockvinodchandranstr}.}. 

\begin{theorem}[\cite{dercconference,dercjournal}]   \label{thm:hitchcockvinodchandran}%
	For every $\F \subseteq \Sigma^\infty$, $\Kpoly (\F) \leq \cdimp(\F)$.
\end{theorem}

The question whether the reverse inequality holds has remained elusive. Hitchcock and Vinodchandran \cite{dercconference,dercjournal} and later Stull in \cite{stullsurvey} posed the following open question:

\begin{question}[\cite{dercconference,dercjournal,stullsurvey}]
	\label{ques:sequencequestion}
	Is it true that, for every sequence $X \in \Sigma^\infty$
	\begin{align*}
		\cdimp(X) = \Kpoly(X) \;?
	\end{align*}
	
\end{question}

\subsubsection*{Polynomial-time Strong Dimension}

%
%

\textbf{Polynomial-time strong dimension} is a dual notion of polynomial time dimension \cite{athreya2007effective,stullsurvey}. As in the case of polynomial time dimension, there exist two notions, one defined using $s$-gales and the other defined in terms of time bounded Kolmogorov complexity. An $s$-gale $d$ strongly succeeds on an infinite sequence $X$ if  $\liminf\limits_{n \to \infty} d(X_1X_2\dots X_n) = \infty$.

\begin{definition*}[Polynomial-time strong dimension \cite{athreya2007effective,stullsurvey}]
	The polynomial-time dimension of $ \F \subseteq \Sigma^\infty$ is defined as 
	\begin{align*}
		\cDimp(\F) = \inf \{s \mid \exists \text{ a polynomial-time  } s\text{-gale } d \text{ such that d strongly succeeds on all } X \in \F \}.
	\end{align*}
	For a sequence $X \in \Sigma^\infty$, define $ \cDimp(X) = \cDimp(\{X\})$.
\end{definition*}
The strong dimension analogue of $\Kpoly$ is defined by replacing the $\liminf$ in the definition of $\Kpoly$ with $\limsup$.
\begin{definition*}[\cite{dercconference,dercjournal}]
	For any $\F \subseteq \Sigma^\infty$,
	\begin{align*}
		\Kstrpoly(\F) = \inf\limits_{t \in \poly}\sup\limits_{X \in \F}\limsup\limits_{n \rightarrow \infty} \frac{K_t(X \restr n)}{n}.
	\end{align*}
\end{definition*}

Similar to the conclusion of Theorem \ref{thm:hitchcockvinodchandran}, $\cDimp$ is always greater than or equal to $\Kstrpoly$.

\begin{theorem}[\cite{dercconference,dercjournal}]   \label{thm:hitchcockvinodchandranstr}%
	For every $\F \subseteq \Sigma^\infty$,  $\Kstrpoly (\F) \leq \cDimp(\F)$.
\end{theorem}

Stull in \cite{stullsurvey} posed the following question:  \emph{Are the notions of polynomial time strong dimension formulated using $s$-gales ($\cDimp$) and time bounded Kolmogorov complexity ($\Kstrpoly$) equivalent?}.

\subsubsection*{Polynomial time samplable distributions on infinite sequences}
We generalize the notion of polynomial time samplable distributions \cite{ilango2022robustness,hirahara2023duality,liu2020one} to probability distributions on the space of infinite sequences (see Section \ref{subsec:ptimesamplabledefinition} for formal definitions).

Let $\nu$ be a probability measure over $\Sigma^\infty$. For every $n$, let $\nu_n(w) = \nu(X \in \Sigma^\infty : w \prefix X)$\footnote{$\prefix$ denotes the prefix operator} denote the probability distribution induced by $\nu$ on $\Sigma^n$. Now, we define polynomial time samplable distributions on $\Sigma^\infty$.

\begin{restatable}{definition}{ptimsampldistn}
	A probability distribution $\nu$ over $\Sigma^\infty$ is polynomial time samplable if there exists a probabilistic Turing machine $M$ that uses $q(n)$ random bits, where $q$ is a polynomial, such that for every $n$ and $w \in \Sigma^n$, $\pr_{r \sim \Sigma^{q(n)}}[M(1^n,r)=w]=\nu_n(w)$\footnote{An alternate way of defining polynomial time samplable distributions on $\Sigma^\infty$ is the following: $\nu=\{\nu_n\}$ is a \emph{polynomial time samplable distribution on $\Sigma^\infty$} if $\{\nu_n\}$ is a \emph{polynomial time samplable distribution} (in the sense of \cite{ilango2022robustness,hirahara2023duality,liu2020one}) such that for any $w \in \Sigma^n$, $\nu_{n}(w) = \nu_{n+1}(w.0) + \nu_{n+1} (w.1)$. The equivalence of these notions follows using routine measure theoretic arguments.}. 
\end{restatable}

One-way functions are secure against inversion by probabilistic polynomial time adversaries. However, polynomial time gales are defined in terms of computation using deterministic machines. In order to bridge this gap, we define martingales and $s$-gales that are approximable using probabilistic polynomial-time machines. 

\begin{restatable}{definition}{ptimapproxdistn}
	Let \( d : \Sigma^* \to [0, \infty) \cap \Q \) be an $s$-supergale and $\nu$ be any probability distribution over $\Sigma^\infty$. $d$ is $t(n)$-time $\nu$-approximable  if for every constant $k$ there exist a probabilistic $t(n)$-time machine $M$ 
	and constant $c<1$  such that for every $n$, $\{w\in \Sigma^n :M(w) \not\in [c\cdot d(w),d(w)]\} \subseteq \mathrm{supp}(\nu_n)$ and $\nu_n\{w \in \Sigma^n :M(w) \not\in [c\cdot d(w),d(w)]\} \leq n^{-k}$.\end{restatable}

In the above, $\mathrm{supp}(\nu_n)$ denotes the support of the distribution $\nu_n$, i.e. $\mathrm{supp}(\nu_n) = \{w \in \Sigma^n : \nu_n(w) > 0\}$. A supergale $d$ is $\nu$-approximable if the set of strings at which the algorithm $M$ makes an error in approximating $d(w)$ up to a constant multiplicative factor lies within the support of $\nu_n$, and the measure of the set of such strings according to $\nu_n$ is bounded by an inverse polynomial function.

\subsection{Our Results}

\subsubsection{One way functions and Dimension gaps}

Our main result (Theorem \ref{thm:equivalencetheorem}) shows that the existence of one-way function implies the existence  of uniform dimension gaps between	$\cdimp$ and $\Kpoly$ over a collection of sequences in $\Sigma^\infty$, such that the collection has probability 1 according to a short seed polynomial time samplable distribution. Futhermore we show that such a distrution implies existence of infinitely often one way functions.

\begin{restatable}{theorem}{equivalencetheorem} \label{thm:equivalencetheorem}
	Suppose that one-way functions exist. Then, for every $s < \frac{1}{2}$, there exists a polynomial-time samplable distribution $\nu$ over $\Sigma^\infty$ such that:
	\begin{enumerate}
		\item The number of random bits used by the sampler for $\nu$ on input $1^n$ is at most $sn$\footnote{See the remarks at the end of section \ref{sec:maintheoremintro} regarding dimension gaps and short seed distributions.}.
		\item For every $s' \in (s, \frac{1}{2})$ and every polynomial-time $\nu$-approximable $s'$-supergale $d$, we have
		\[
		\nu(S^\infty(d)) = 0.
		\]
	\end{enumerate}
	Furthermore, this implies the existence of infinitely-often one-way functions\footnote{See the remarks in section \ref{sec:ioowfs} regarding infinitely-often one-way functions and polynomial time dimension and also Lemma \ref{lem:equivalenceowfconverse}.}.
\end{restatable}

%

The existence of polynomial-time samplable distributions satisfying the condition 1 implies that, with probability $1$ according to $\nu$, the polynomial-time Kolmogorov complexity $\Kpoly(X)$ of $X \in \Sigma^\infty$ is at most $s$. We later show in the proof of Lemma \ref{thm:thmforSeqpsrg} that the polynomial-time constructive dimension $\cdimp(X)$ of almost every such sequence is at least $1/2$. Since $s < 1/2$, this establishes the existence of a dimension gap on a probability 1 set according to $\nu$.

\subsubsection{Applications}

As an application of Theorem \ref{thm:equivalencetheorem}, we show that if one-way functions exist then equality does not hold between $\cdimp$ and $\Kpoly$ for sets $\F \subseteq \Sigma^\infty$. To show this, we prove the contrapositive statement that if the equality holds, then one-way functions do not exist. 

\begin{restatable}{theorem}{thmforSetspsrg} \label{thm:forSetspsrg}%
	If for all $\F \subseteq \Sigma^\infty$, $$\cdimp(\F) = \Kpoly(\F),$$
	then one-way functions do not exist.
\end{restatable}

We then show the stronger result that if one-way functions exist then then equality does not hold between $\cdimp$ and $\Kpoly$ for sequences $X \in \Sigma^\infty$.

\begin{restatable}{theorem}{thmforSeqpsrg} \label{thm:thmforSeqpsrg}
	If for all $X \in \Sigma^\infty$, $$\cdimp(X) = \Kpoly(X),$$
	then one-way functions do not exist.
\end{restatable}

Thus, conditioned on the existence of one-way functions, we demonstrate the existence of sequences $X$ such that $\Kpoly (X) < \cdimp(X)$. 

In the last section of the paper we demonstrate how the proof of Theorem \ref{thm:thmforSeqpsrg} can be extended to show that if one-way functions exist, then the distance between these quantities can be arbitrarily close to the maximum possible value of $1$. 

\begin{restatable}{theorem}{gaptheoremsequences}
	\label{thm:gaptheoremsequences}
	If one-way functions exist, then for any $\epsilon>0$, there exists $ X \in \Sigma^\infty$ such that,
	\begin{align*}
		\cdimp(X) - \Kpoly (X) \geq 1-\epsilon.
	\end{align*}
\end{restatable}

For polynomial time strong dimension, we show that if one-way functions exist, then there exist  sequences for which the gap between $\Kstrpoly$ and $\cDimp$ is arbitrarily close to $1$.

\begin{restatable}{theorem}{strongdimensiongapsequences}
	\label{thm:strongdimensiongapsequences}
	If one-way functions exist then for any $\epsilon>0$, there exists $ X \in \Sigma^\infty$ such that,
	\begin{align*}
		\cDimp(X) - \Kstrpoly(X) \geq 1-\epsilon.
	\end{align*}
\end{restatable}

Therefore, we show that if one-way functions exist, the answers to both parts of Open Question 3.71 from \cite{stullsurvey} are negative.

\subsection{Techniques}
 

\textbf{Proof of Theorem \ref{thm:equivalencetheorem}}: We use a new construction to extend pseudorandom generators to obtain short seed polynomial-time samplable distributions on infinite sequences (Section \ref{sec:constructionofg}, Figure \ref{fig:gConstrn}).  We use existing standard $s$-gale manipulation techniques (Lemma \ref{lem:galetofastgrowingmgale} and Lemma 
\ref{lem:optimalmartingaleconstruction}) to convert winning gales on outputs of short seed distributions.  Using these we show a new technique to build a distinguisher algorithm (Algorithm \ref{alg:distinguisher}) that breaks the pseudorandom generators used. To show that the distinguisher succeeds with high probability on pseudorandom outputs, we use a argument involving the Borel-Cantelli lemma (Lemma \ref{lem:BorelCant}) to transform the occurrence of infinitely many events into the occurrence over a large fraction, as required for the distinguisher. We use the Kolmogorov Inequality (Lemma \ref{lem:KolmogorvInequality}) to show that the distinguisher succeeds with low probability on uniformly random outputs.

For the converse, we use the theorem on \emph{universal extrapolation} by Ilango, Ren and Santhanam (see Theorem 20 from \cite{ilango2022robustness}) along with a gale construction technique.

\medskip 

\textbf{Proof of Theorem \ref{thm:thmforSeqpsrg}}: Non robustness of polynomial-time dimension for sets (Theorem \ref{thm:forSetspsrg}) follows as a corollary of Theorem \ref{thm:equivalencetheorem}, along with some compactness arguments. However, to extend the argument to show non-robustness of polynomial-time dimension for sequences (Theorem \ref{thm:thmforSeqpsrg}) requires additional technical effort. We provide a new gale combination technique (Lemma \ref{lem:Mgalecombination}) to construct a $t(n).n.\log(n)$ $s$-gale that is universal over all $t(n)$-time $s$-gales. We use this along with a measure-based partitioning trick to construct a polynomial time $s$-gale that is universal over all polynomial time $s$-gales (up to a positive measure subset)  to prove Theorem \ref{thm:thmforSeqpsrg}.

\subsection{Proof Outline}

In this section we give an informal account of the proofs of the main results. The full proofs are given in later sections.

\subsubsection*{Dimension gaps from One-Way Functions}

We first show that one-way functions imply the existence of distributions over which dimension gap exists almost everywhere.
Precisely, we show that if one-way functions exist, for any $s < 1/2$, there are polynomial time samplable distribution $\nu$ over $\Sigma^\infty$ samplable using $sn$ bits such that for any $s' \in (2s,1/2)$, any polynomial time $s'$-supergale that is $\nu$-approximable  can succeed only on a ``small" $\nu$-measure $0$ subset of $\Sigma^\infty$. See Section \ref{sec:forwardimplication} for the formal proof.

\begin{restatable}{lemma}{equivalenceforward}
	\label{lem:equivalenceforward}
	If one-way functions exist, then for every $s<1/2$, there exist a polynomial time samplable distribution $\nu$ over $\Sigma^\infty$ such that:
	\begin{enumerate}
		\item The number of random bits used by the sampler for $\nu$ on input $1^n$ is at most $sn$.
		\item For every $s' \in (s,1/2)$ and any polynomial-time $\nu$-approximable $s'$-supergale $d$, $\nu(S^\infty(d))=0$.
	\end{enumerate}
\end{restatable}

\textbf{Proof Outline}. We start with the assumption that one-way functions exist. This implies the existence of pseudorandom generators $\{G_n\}_{n \in \N}$ running in time $t(n) \in \poly(n)$ mapping strings of length $sn$ to strings of length $n$ for every $s<1$. For the sake of convenience,  take any $s$ such that $s = 2^{-m}$ for some $m \in \N$, $m > 1$.
We first \emph{extend} the PRG $\{G_n\}_{n \in \N}$ to a mapping between infinite sequences, $g:\Sigma^\infty \to \Sigma^\infty$. We then use $g$ to construct a polynomial time samplable distribution $\nu$ on $\Sigma^\infty$.

The mapping $g$ is constructed as follows. We first define mappings $g_k$ from $\Sigma^k$ to $\Sigma^*$ such that $k = s.2^n$ for some $ \in \N$.
On an input $w \in \Sigma^k$, where $g_k$ first divides $w$ into blocks $x_n$ of size $s.2^i$ for every $i < n$. Then, for every block $x_i$ in $w$,  $g$ applies the appropriate function from $\{G_n\}_{n \in \N}$ to obtain an output block $y_i$ of length $2^i$. Now $g_k(w)$ is defined as the concatenation of blocks $y_i$ for every $i \leq n$. Now, for any $X \in  \Sigma^\infty$ we define the function $g$ as the \emph{limit} of the mappings $g_k$, $g(X)=\lim_{k \to \infty} g_k(X \restr k)$. We abuse the notation and use $g$ instead of $g_k$ when the context is clear. Figure \ref{fig:gConstrn} gives a detailed illustration of the construction of $g$.

 The mapping $g:\Sigma^\infty \to \Sigma^\infty$ naturally generates a \emph{short seed} polynomial-time samplable distribution $\nu$ over $\Sigma^\infty$. The sampling algorithm $M$ for $\nu$ on input $1^n$, maps a random seed of length $s\cdot 2^{\lceil \log_2(n) \rceil}$ using $g$ to a string of length $2^{\lceil \log_2(n) \rceil}$, and thereafter trims the output to $n$ bits. The number of random bits used by $M$ on input $1^n$ is at most $2s \cdot n$. 
 
 \begin{figure}[] \label{fig:gConstrn}
 	\centering
 	\includegraphics[width=0.75\linewidth]{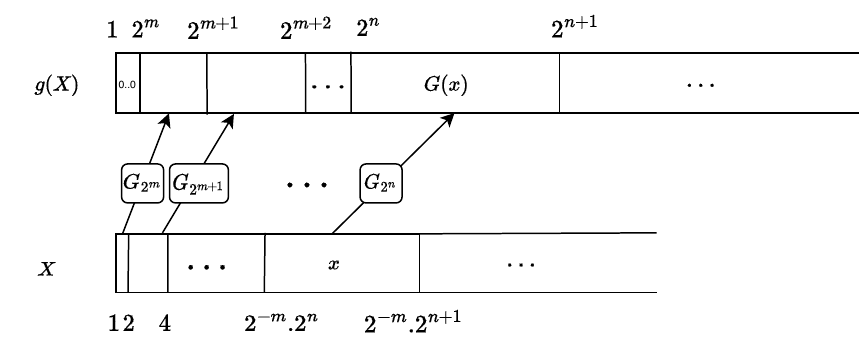}
 	\caption{Illustration of the construction of $g(X)$. The first $s = 2^m$ bits of the $g(X)$ are 0. Thereafter, for every for every $n>m$, the block $X[s.2^{{n-1}}, s.2^{n}-1]$ is mapped to $G_{2^{n-1}}(X[s.2^{{n-1}}, s.2^{n}-1])$ in the output string $g(X)$.}
 	\label{fig:example}
 \end{figure}

Towards building a contradiction, let there exist some $s' \in (2s,1/2)$ and a polynomial-time $\nu$-approximable $s'$-supergale $d'$ such that $\nu(S^\infty(d'))>0$. We show that $d'$ can be used to break the PRG $\{G_n\}_{n \in \N}$, obtaining a contradiction.

Using standard techniques on gales, we convert the polynomial time $\nu$-approximable $s'$-supergale $d'$ to a polynomial-time $\nu$-approximable martingale $\tilde{d}$ that gains a \emph{significant} amount of money over blocks of length $2^n$ for infinitely many $n$. More precisely, for some $\tilde{s}\in (2s',1)$, for all $Y \in g(\Sigma^\infty$), there exist infinitely many $n \in \N$ such that
\begin{align} 
	\label{eq:fastmartingalecondition}
	\tilde{d}(Y \restr 2^{n+1}) > 2^{(1-\tilde{s})2^{n}} \tilde{d}(Y \restr 2^{n}).
\end{align}

Let $M$ be the Turing machine that $\nu$-approximates $\tilde{d}$. That is for all $n \in \N$, $\nu_n\{w \in \Sigma^n :M(w) \not\in [\constc \cdot d(w),d(w)]\} \leq n^{-k}$ where $\constc \in \Q,k \in \N$ are constants such that $\constc \leq 1$ and $k \geq 1$. We use $M$ to build a distinguisher algorithm $A$ that breaks the PRG $\{G_n\}$.
 We describe the behavior of $A$ on inputs of length $2^n$. Let $w$ be an input of size $2^n$. The distinguisher randomly chooses $s.2^n$ bits $r$ and takes $w' = g(r)$. Now $A$ outputs $1$ if and only if: 
\begin{align*}
	M(w'w) \geq \constc \cdot 2^{(1 - \tilde{s})\lvert w \rvert} \cdot M(w').
\end{align*}
Otherwise $A$ outputs $0$. Since $M$ is polynomial time computable, $A$ is a polynomial time algorithm. 

To show that $A$ breaks the PRG $\{G_n\}_{n \in \N}$, it suffices to show that 
for infinitely many $n \in \N$, taking $N = 2^n$, $|\pr_{x \sim \Sigma^{s.N}}[A(g(x)) = 1] - \pr_{w \sim \Sigma^N}[A(w) = 1]| > N^{-c}$ for some constant $c$.

We first analyse the behaviour of $A$ on inputs of the PRG $\{G_n\}$.
We have shown that $\tilde{d}$ satisfies the condition (\ref{eq:fastmartingalecondition}) over $Y \in g(\Sigma^\infty)$ for infinitely many $n$. But the \emph{fraction} of strings of a particular length $n$ satisfying (\ref{eq:fastmartingalecondition}) may in fact be negligible. In order to overcome this difficulty, we use an argument involving the Borel Cantelli Lemma to show that there exist infinitely many $n \in \N$, such that over at least a $1/n^2$ fraction of strings $w \in \Sigma^{s.2^n}$, 
\begin{align} \label{eg:AftreBorelCantelli}
	\tilde{d}(g(w)) > 2^{(1-\tilde{s})2^{n}} \tilde{d}(g(w) \restr 2^{n-1})).	
\end{align}

Let $n+1$ be one of the lengths at which the condition \ref{eg:AftreBorelCantelli}  holds. 
Let $w,w' \in \Sigma^{2^n}$ such that $w = G_{2^n}(x)$ for some $x \in \Sigma^{s.2^{n}}$ and $w' = g(r)$ for some $r \in \Sigma^{s.2^{n}}$. 

In order for the condition $M(w'w) \geq \constc \cdot 2^{(1 - \tilde{s})\lvert w \rvert} \cdot M(w')$ to hold, it suffices that $\tilde{d}(w'w) \geq 2^{(1 - \tilde{s})\lvert w \rvert} \cdot \tilde{d}(w')$ and $M(w'w) \geq \constc \cdot \tilde{d}(w'w)$ and $M(w) \leq \tilde{d}(w)$. From the argument given above, it follows that the first condition holds with probability at least $(n+1)^{-2}$.
Since $\tilde{d}$ is $\nu$-approximable, it follows that the second and third conditions hold  with probability at least $1-2^{-(kn-1)}$ for some constant $k$.

Therefore it follows that for infinitely many n,
\begin{align}\label{eq:AnalysisPO1}
	\pr_{(x,r) \sim \Sigma^{s \cdot 2^n} \times \Sigma^{s \cdot 2^n}}[A(G_{2^n}(x),r)) = 1] &\geq \frac{1}{(n+1)^2}-\frac{1}{2^{kn}}.
\end{align}

We now analyze the behaviour of $A$ on uniformly random inputs. We use the Kolmogorov inequality for supermartingales to show that this is inverse exponentially small in input size.

We need an upper bound on the number of strings $w,r \in \Sigma^{2^n}$ such that $A(w,r) = 1$. This happens when $M(w'w) \geq \constc \cdot 2^{(1 - \tilde{s})\lvert w \rvert} \cdot M(w')$.  Using the Kolmogorov inequality for martingales, we show that for a fixed $w'$, the fraction of strings $w \in \Sigma^{2^{n}}$ such that $\tilde{d}(w'w) \geq \constc^2 \cdot 2^{|w|(1-\tilde{s})}.\tilde{d}(w')$ is less than $ \constc^{-2} \cdot 2^{-|w|(1-\tilde{s})}$. 
 For the remaining set of strings, we have that $\tilde{d}(w'w) < \constc^2 \cdot 2^{2^{n}(1-\tilde{s})}. \tilde{d}(w')$. If we have that $M(w'w) \leq \tilde{d}(w'w)$ and $M(w') \geq \constc \cdot \tilde{d}(w')$, it follows that $A(w,r) = 0$. Using the fact that $\tilde{d}$ is $\nu$-approximable, we show the following that for some constant $k$, for all $n \in \N$, taking $N = 2^n$,
\begin{align} \label{eq:AnalysisPO2}
	\pr_{w \sim \Sigma^N} \pr_{r \in \Sigma^{s.N}} [A(w,r) = 1] <\frac{1}{\constc^2 \cdot 2^{ N  ( 1- \tilde{s})}}+\frac{1}{N^k}+\frac{1}{2^{ N ( 1- {s})}}.
\end{align}

From Equation \ref{eq:AnalysisPO1} and \ref{eq:AnalysisPO2}, it follows that $A$ breaks the PRG $\{G_n\}$, which is a contradiction.


\subsubsection*{Infinitely often one-way functions from dimension gaps} 

We now give a proof outline that the existence of distributions over which dimension gap exists almost everywhere implies existence of infinitely often one-way functions. See Section \ref{sec:backwardimplication} for the formal proof.

\begin{restatable}{lemma}{equivalenceconverse} \label{lem:equivalenceconverse}
	If for some $s<1$ there exist a polynomial time samplable distribution $\nu$ over $\Sigma^\infty$ such that:
	\begin{enumerate}
		\item The number of random bits used by the sampler for $\nu$ on input $1^n$ is at most $sn$.
		\item There exist some $\hat{s} \in(s,1]$ such that for every $s' \in (s,\hat{s})$ and any polynomial-time $\nu$-approximable $s'$-supergale $d$, $\nu(S^\infty(d))=0$.
	\end{enumerate}
	Then, infinitely-often one-way functions exist.
\end{restatable}

\textbf{Proof Outline}. We start with the hypothesis that for some $s<1$ there exist a polynomial time samplable distribution $\nu$ over $\Sigma^\infty$ such that the number of random bits used by the sampler for $\nu$ on input $1^n$ is at most $sn$ and there exists $\hat{s} \in (s,1]$ such that for every $s' \in (s,\hat{s})$ and polynomial-time $\nu$-approximable $s'$-supergale $d$, $\nu(S^\infty(d))=0$. Assume that there exist no infinitely often one-way functions.

Fix arbitrary $s'$ such that $s<s'<\hat{s}$ and $2^{s'}\in \Q$. Consider the function $d:\Sigma^* \to [0,\infty)$ defined as $d(w)=2^{s'\lvert w \rvert} \nu(w)$. Using the fact that $\nu$ is a short-seed polynomial time samplable distribution, we show that $d$ defined above satisfies the following property,
	\begin{align*}
		\nu\left\{X \in \Sigma^\infty: \limsup_{n \to \infty} d(X \restr n) = \infty\right\}=1.	
	\end{align*}

We now show that if infinitely-often one-way functions do not exist, there exists an algorithm $M$ that $\nu$-approximates $d$. Using a theorem on \emph{universal extrapolation} by Ilango, Ren and Santhanam (see Theorem 20 from \cite{ilango2022robustness}), it follows that for any $q \geq 1$, there exists a probabilistic polynomial time algorithm $\mathcal{A}$ and constant $c<1$ such that for all $n$,
\begin{align*}
\pr_{w \sim \nu_n} [c \cdot\nu_n(w) \leq \mathcal{A}(w) \leq \nu_n(w)]	\geq 1-O\left(\frac{1}{n^q}\right).
\end{align*}
 Furthermore, since there infinitely-often one-way functions do not exist (and therefore weak infinitely-often one-way functions do not exist), there is an algorithm $\mathcal{I}$ which \emph{inverts} the function $f:\Sigma^* \to \Sigma^*$ defined as $$f(x)=M(1^{\lvert x \rvert^{1/c'}},x)$$ on all but finitely many input lengths $n$. We combine the algorithms $\mathcal{A}$ and $\mathcal{I}$ to give an algorithm $M$ satisfying the following properties: for every $n$, $\{w:M(w) \not\in [c\cdot d(w),d(w)]\} \subseteq \mathrm{supp}(\nu_n)$ and
\begin{align*}
\pr_{x \sim \nu_n} [c \cdot d(x) \leq M(x) \leq d(x)]	\geq  1-O\left(\frac{1}{n^q}\right).
\end{align*}
 Using the fact that $\mathcal{A}$ and $\mathcal{I}$ are randomized polynomial time algorithms and $q$ is arbitrary we show that $M$ approximates $\nu$ in polynomial time. Since this contradicts our hypothesis, it must be the case that there exist infinitely-often one-way functions.

\subsubsection*{Dimension gaps over sequences from one-way functions}

We use Lemma \ref{lem:equivalenceforward} to show that if one-way functions exist then there exists sequences $X$ for which $\Kpoly(X)$ and $\cdimp(X)$ are separated by a gap of $1/2$. Thereafter we refine our proof to increase the gap to any value arbitrarily close to $1$. The corresponding theorems for sets is trivially implied by considering the set $\F=\{X\}$.

\thmforSeqpsrg*

\textbf{Proof Outline.}	We start with the assumption that one-way functions exist. From Lemma \ref{lem:equivalenceforward}, for every $s<1/2$ there exist a polynomial time samplable distribution $\nu$ over $\Sigma^\infty$ and polynomial $t$ such that:
\begin{enumerate}
\item The number of random bits used by the sampler for $\nu$ on input $1^n$ is at most $sn$.
\item For every $s' \in (s,1/2)$ and any polynomial-time $\nu$-approximable $s'$-supergale $d$, $\nu(S^\infty(d))=0$.
\end{enumerate}
Consider the following set $
\mathcal{F}=\bigcap\limits_{n \geq 1} \bigcup\limits_{w \in \mathrm{supp}(\nu_n)} C_w.$
Using an argument involving the compactness of $\Sigma^\infty$ as a topological space, we demonstrate that $\F$ is a non-empty set such that $\nu(\mathcal{F})=1$. Since the number of random bits used by the sampler for $\nu$ on input $1^n$ is at most $sn$, $\Kpoly(X)\leq s$ for every $X \in \F$. For the sake of contradiction, assume that there exist $s' \in (s,1/2)$ such that $\cdimp(X)<s'$ for every $X \in \F$. So, every $X \in \F$ there exists an exact computable $s'$-gale $d'_X$ that runs in time $t_{d'_X}(n) \in \poly(n)$, such that $d'_X$ succeeds on $X$.

 We now use a partitioning trick to build a set $\mathcal{S} \subseteq \Sigma^\infty$ such that $\nu(\mathcal{S})>0$. Thereafter we use a new gale combination technique to construct a $n^{k+c}$-time computable $s'$-gale $d'$ that succeeds on all sequences in $\mathcal{S}$. Since $d$ is polynomial-time computable, $d$ is trivially a polynomial-time $\nu$-approximable $s'$-supergale. Now, since we know that $\nu(S^\infty(d)) \geq \nu(\mathcal{S})>0$, we obtain a contradiction. Hence, it must be the case that there exist a sequence $X\in \F$ with $\cdimp(X)>s'$. Since $s'$ was arbitrary, the theorem follows.

 \gaptheoremsequences*
\textbf{Proof Outline.} In order to demonstrate that for every $\epsilon$, there exist sequences $X$ for which the gap between $\cdimp$ and $\Kpoly$ is greater than $1-\epsilon$, we combine the proof techniques in the proofs of Lemma \ref{lem:equivalenceforward} and the proof of existence of sequences with gaps close to $1/2$ sketched above. Also, we require an important modification in the construction of the martingale $\tilde{d}$ in the proof of Lemma \ref{lem:equivalenceforward}. While constructing the martingale $\tilde{d}$ in the proof of Lemma \ref{lem:equivalenceforward}, we used the assumption that $s'<1/2$. This is important because in order to obtain the condition $\tilde{d}(Y \restr 2^{n+1}) > 2^{(1-\tilde{s})2^{n}} \tilde{d}(Y \restr 2^{n})$ for infinitely many $n$ over any $Y \in g(\Sigma^\infty)$, we require $\tilde{s} > 2s'$. For $\tilde{s}$ to be less than $1$, we require $s'<1/2$. The major tool we need to overcome this hurdle is a generalization of the construction of martingale $\tilde{d}$ to every $s' < 1$ (see Lemma \ref{lem:optimalmartingaleconstruction}). We transform $d'$ into a polynomial time martingale $\tilde{d}$ satisfying the following property: let $\tilde{s}$ and $s''$ be such that $\tilde{s}>s''>s'$. Then for any $Y \in g(\Sigma^\infty)$ there exist infinitely many $n$ satisfying either of the following conditions:
\begin{enumerate}
	\item\label{item:optimalmgale1} $\tilde{d}(Y \restr 2^n) > 2^{(1-\tilde{s})2^{n-1}} \tilde{d}(Y \restr 2^{n-1})$.
	\item\label{item:optimalmgale2} There exists $\ell$ satisfying $ ((\tilde{s}-s'')/s'') \cdot 2^{n-1} \leq \ell \leq 2^{n-1}$ such that $\tilde{d}(Y \restr 2^{n-1} + \ell) > 2^{(1-s'')(2^{n-1}+\ell)}$. 
\end{enumerate}

We need the second condition above to handle the case when $s'$ may be between $1/2$ and $1$. Finally, we modify Algorithm A from the proof of Lemma \ref{lem:equivalenceforward} to incorporate the second condition above. We adapt the analysis of the algorithm and the arguments involving the Borel Cantelli lemma appropriately to show that the modified algorithm $A$ is a distinguisher for the PRG $\{G_n\}_{n \in \N}$.  

\subsection{Related Work}

In \cite{liu2020one}, Liu and Pass proved that
one-way functions exist if and only if the \emph{$t$-time bounded
	Kolmogorov complexity problem} $\mathrm{MK}^t\mathrm{P}$, is
\emph{mildly} hard on average. In a follow-up work \cite{liu2021one}
they define a problem that is $\NP$-complete under randomized
reductions  whose mild average
case hardness is equivalent to existence of one-way functions. Allender, Cheraghchi, Myrisiotis, Tirumala and Volkovich
\cite{allender2021one} showed similar results in the setting of
$\mathrm{KT}$ complexity, basing the existence of one-way functions on the average case hardness of $\mathrm{McKTP}$ (an analogue of
$\mathrm{McK}^t\mathrm{P}$ for $\mathrm{KT}$-complexity).  Hirahara \cite{Hirahara23}
gave the first characterization of a one-way function by worst-case hardness assumptions.  He showed that  one-way function exists if and only if it
is NP-hard to approximate the distributional Kolmogorov complexity under randomized polynomial-time reductions, assuming NP is hard in the worst case.
Hirahara, Lu and Oliveira \cite{HirhLuOliveria} showed the relationship between One-way functions and pKt complexity. Ilango, Ren, and Santhanam \cite{ilango2022robustness} characterize the existence of OWFs by the average-case hardness of approximating Kolmogorov complexity on samplable distributions. Hirahara, Ilango, Lu, Nanashima and Oliveira \cite{hirahara2023duality} gave a complete characterization of one-way
functions in terms of the average case failure of symmetry of
information (and related properties like the conditional coding theorem)
for $\mathrm{pK}^t$.

Considerable research has been devoted to understanding whether
fundamental properties of unbounded Kolmogorov complexity survive in
the time bounded setting. One of the important properties that was
studied in this context is the symmetry of information of Kolmogorov
complexity
\cite{LiVitanyi,ShenUspenskyVereshchagin,Downey10,Nies2009}. Longpr\'e
and Mocas \cite{LM95} (also Longpr\'e and Watanabe \cite{LW92}) showed
that if one-way functions exist then symmetry of information does not
hold for time bounded Kolmogorov complexity. In \cite{Hirahara2022}
and \cite{Goldberg2022Simpler}, symmetry of information for polynomial
time bounded probabilistic Kolmogorov complexity $\mathrm{pK}^t$ (see
\cite{GKLO22,Lu2022Theory}) was derived from the stronger
assumption $\mathrm{DistNP} \subseteq \mathrm{AvgBPP}$.

This work explores whether the characterization of constructive dimension in terms of Kolmogorov complexity \cite{LutzDimension2003,Mayordomo02} extends to the time-bounded setting. Our results establish a duality between the separation of $\cdimp$ and $\Kpoly$ as notions of polynomial-time dimension over a \emph{sufficiently large collection} of sequences and the existence of one-way functions.

As an application of our main result, we show that if one-way functions exist then
polynomial time bounded Kolmogorov complexity does not yield a
characterization of polynomial time dimension of infinite sequences,
thereby refuting an open question posed by Hitchcock and Vinodchandran
in \cite{dercconference} and by Stull in \cite{stullsurvey}. 

\subsection{Open Problems}
Our findings suggest several promising directions for future research. A key open question is whether weaker hardness assumptions, such as $\P \neq \NP$, $\mathrm{DistNP} \not\subseteq \mathrm{AvgP}$ or $\mathrm{DistNP} \not\subseteq \mathrm{AvgBPP}$  can be characterized in terms of the separation of polynomial-time dimension, as defined using time-bounded $s$-gales and time-bounded Kolmogorov complexity in appropriate settings. Specifically,\emph{ can these hardness assumptions be equivalently expressed through if-and-only-if statements involving polynomial-time dimension?}.

Our results reveal a duality between the non-robustness of polynomial-time dimension and the existence of one-way functions. An interesting open question is whether notions of polynomial-time randomness-defined via betting algorithms and Kolmogorov incompressibility of prefixes (see \cite{stullsurvey} for more details) capture equivalent concepts of randomness. Polynomial-time betting algorithms that succeed on specific sequences do not necessarily provide guarantees on the rate of capital accumulation in general. Consequently, investigating the robustness of polynomial-time randomness and its potential connections to cryptographic and complexity-theoretic primitives may require fundamentally new approaches.

\section{Preliminaries}
\label{sec:preliminaries}
In this section, we introduce the basic notation and formally define the necessary concepts from the theory of pseudorandomness and constructive dimension.

\subsection{Notation}
Let $\Sigma$ denote the binary alphabet $\{0,1\}$. $\Sigma^n$ denotes the set of $n$-length strings and $\Sigma^*$ denotes the set of all finite strings over $\Sigma$. For a finite string $w \in \Sigma^*$, $\lvert w \rvert$ denotes the length of $w$. $1^*$ denotes the set of all finite strings in the unary alphabet $\{1\}$. $\Sigma^\infty$ denotes the \emph{Cantor space}, or the set of all infinite strings over $\Sigma$. We use small letters $w$, $x$ to represent finite strings and capital letters $X$, $Y$ to represent infinite strings over $\Sigma$.
For an infinite string $X \in \Sigma^\infty$, $X[i]$ denotes the $i$\textsuperscript{th} bit of $X$ and $X \restr n$ denotes the first n bits of $X$. For infinite string $X=X_1 X_2 X_3 X_4 \dots$ and $m \geq n$, $X[n:m]$ denotes the the substring $X_n X_{n+1} X_{n+2} \dots X_{m}$. For a finite string $w \in \Sigma^*$, $w[i]$, $w \restr n$ and $w[n:m]$ are defined similarly. $\lambda$ denotes the empty string. $\N$ denotes the set of natural numbers and $\Q$ denotes the set of rationals. All the logarithms in the paper are taken to base $2$ unless specified otherwise. For any function $f: \Sigma^\infty \to \Sigma^\infty$ and $\mathcal{S} \subseteq \Sigma^\infty$ let $f(\mathcal{S})$ denote $\{f(X): X \in \mathcal{S}\}$.

For $w \in \Sigma^*$, $C_w$ denotes the \emph{cylinder set} of $w$, that is the set of all infinite sequences in \( \Sigma^\infty \) that begin with \( w \). Let $\mathcal{B}(\Sigma^\infty)$ denote the Borel $\sigma$-algebra over $\Sigma^\infty$ generated by the set of all cylinder sets $\{C_w: w\in \Sigma^*\}$. For two finite strings $x,y \in \Sigma^*$, $x.y$ denotes the concatenation of the strings $x$ and $y$.

We use $\poly(n)$ to denote the set of polynomial (time) functions $\bigcup_{c \in \N}\{n^c\}$. When the context is clear, we also use $\poly$ to denote $\poly(n)$. We use $U_n$ to denote the uniform probability distribution over $\Sigma^n$. 

\subsection{Time Bounded Kolmogorov complexity}
We define time bounded Kolmogorov complexity.

\begin{definition}[$t$-time bounded Kolmogorov complexity \cite{Kol65,LiVitanyi}]
	Let \( t \) be a time-constructible function, and \( x \in \Sigma^* \) be a finite binary string. The \( t \)-time bounded Kolmogorov complexity of \( x \)  is
	\[
	K_t(x ) = \min\{|\Pi| \text{ such that } \Pi \in \Sigma^* \text{ and } \mathcal{U}(\Pi) = x \text{ in } t(|x|) \text{ steps}\}.
	\]
	where $ \mathcal{U}$ is any fixed universal Turing machine.
\end{definition}

Since in the study of resource dimension we divide $K_t(x)$ by $\lvert x \rvert$, our results are unaffected by the choice between prefix-free Kolmogorov complexity and plain Kolmogorov complexity. We use plain Kolmogorov complexity in the rest of the paper. 

\subsection{Polynomial-time dimension}

We first give the formal definition of polynomial time dimension using polynomial time $s$-gales. Towards this end, we first define the concepts of $s$-gales, supergales, martingales and supermartingales.
\begin{definition}[\cite{LutzDimension2003,diss}]
	For $s \in \Rplus$,
	\begin{enumerate}
		\item\label{item:sgale} An $s$-supergale is a function $d: \Sigma^{*}
		\to \Rplus$ such that $d(\lambda) < \infty$ and for every finite string \( w \in \Sigma^* \),
		\[ d(w) \geq \frac{d(w0) + d(w1)}{2^s}.\]
		\item A supermartingale is a $1$-supergale.
		\item An $s$-gale is a function $d: \Sigma^{*}
		\to \Rplus$ with $d(\lambda) < \infty$ for every finite string \( w \in \Sigma^* \),
		\[ d(w) =\frac{d(w0) + d(w1)}{2^s}.\]
		\item A martingale is a $1$-gale.
	\end{enumerate}	
	
\end{definition}

Unless specified otherwise, for every $s$-supergale $d$, we assume that the initial capital $d(\lambda)=1$.

\begin{definition}[\cite{LutzDimension2003,diss}]
	The \emph{success set} of an $s$-supergale $d$ is 
	\begin{align*}
		S^\infty (d) = \left\{ X \in
		\N^\infty \mid \limsup \limits_{n\to\infty} d(X \restr n) =
		\infty\right\}.
	\end{align*}
\end{definition}

\begin{definition}[\cite{LutzDimension2003,diss}]
	The \emph{strong success set} of an $s$-supergale $d$ is 
	\begin{align*}
		\Sstr^\infty (d) = \left\{ X \in
		\N^\infty \mid \liminf \limits_{n\to\infty} d(X \restr n) =
		\infty\right\}.
	\end{align*}
\end{definition}


Now, we define time bounded supergales. We first define $t(n)$-time computable $s$-gales.

\begin{definition}[$t(n)$-time computability \cite{LutzDimension2003}]
	An $s$-supergale \( d : \Sigma^* \to [0, \infty) \cap \Q \) is  $t(n)$-time computable  if there exist a $t(n)$-time computable function $f:\Sigma^* \times 1^* \to \Q$ such that $\lvert d(w)-f(w,1^n)\rvert \leq 2^{-n}$ for every $w \in \Sigma^*$ and $n \in \N$.
\end{definition}

Another notion of time bounded computability is that of exact $t(n)$-time computability.

\begin{definition}[Exact $t(n)$-time computability \cite{LutzDimension2003}]
	An $s$-supergale \( d : \Sigma^* \to [0, \infty) \cap \Q \) is exact $t(n)$-time computable  if there exist a $t(n)$-time computable function $f:\Sigma^* \to \Q$ such that $d(w)=f(w)$ for every $w \in \Sigma^*$.
\end{definition}

Now, we define polynomial-time dimension and strong dimension.

\begin{definition}[Polynomial-time dimension \cite{LutzDimension2003}]
	The polynomial-time dimension of \( \F \subseteq \Sigma^\infty\) is defined as \[\cdimp(\F) = \inf \{s \mid (\exists k) \text{ there is an exact } n^k\text{-time computable } s\text{-gale } d \text{ such that } \F \subseteq S^\infty(d)\}.\] 
	For a sequence $X \in \Sigma^\infty$, define $ \cdimp(X) = \cdimp(\{X\})$.
\end{definition}

\begin{definition}[Polynomial-time strong dimension \cite{athreya2007effective,stullsurvey}]
	The polynomial-time strong dimension of \( \F \subseteq \Sigma^\infty\) is defined as \[\cDimp(\F) = \inf \{s \mid (\exists k) \text{ there is an exact } n^k\text{-time computable } s\text{-gale } d \text{ such that } \F \subseteq \Sstr^\infty(d)\}.\] 
	For a sequence $X \in \Sigma^\infty$, define $ \cDimp(X) = \cDimp(\{X\})$.
\end{definition}

We now make an important remark. Polynomial time dimension and strong dimension are often defined in terms of $n^k$-time computable $s$-gales that need not be exact computable. But these definitions are in fact equivalent. This is a direct consequence of the following lemma from \cite{LutzDimension2003} which is also useful in proving the main results in this paper.

\begin{lemma}[Exact Computation Lemma \cite{LutzDimension2003}]
	\label{lem:LutzExactComputation}
	If there is a polynomial time computable $s$-gale $d$ such that $2^s \in \Q$, then there exists an exact polynomial time computable $s$-gale $\tilde{d}$ such that $S^{\infty}(d) \subseteq S^{\infty} (\tilde{d})$.
\end{lemma}

We will often use the terminology  polynomial-time computable gales to refer to exact $t(n)$-time computable $s$-gales where $t(n) \in \poly(n)$. Similar comment applies in the case of exact polynomial-time computability and other notions related to gales such as martingales and supergales.

Now, we define the asymptotic polynomial-time density of information of a subset $\F \subseteq \Sigma^\infty$. 

\begin{definition}[$\Kpoly$ and $\Kstrpoly$ \cite{dercjournal,stullsurvey,hitchcockthesis}]
	For any $\mathcal{F} \subseteq \Sigma^\infty$, define
	\begin{align*}
		\Kpoly (\F)=\inf_{t \in poly} \sup_{\S \in \F} \liminf_{\nifty} \frac{K_t(\S \restr n)}{n}. \\
		\Kstrpoly(\F)=\inf_{t \in poly} \sup_{\S \in \F} \limsup_{\nifty} \frac{K_t(\S \restr n)}{n}.
	\end{align*}
\end{definition}

For individual sequences, we define $\Kpoly(X)=\Kpoly(\{X\})$ and $\Kstrpoly(X)=\Kstrpoly(\{X\})$. Note that $\Kpoly$ and $\Kstrpoly$ are equivalently defined in the following way.

\begin{definition}[$\Kpoly$ and $\Kstrpoly$ \cite{dercjournal,stullsurvey,hitchcockthesis}]
	For any $X \in \Sigma^\infty$, define
	\begin{align*}
		\Kpoly (X)=\inf_{t \in poly}  \liminf_{\nifty} \frac{K_t(\S \restr n)}{n}
		\text{ and }
		\Kstrpoly(X)=\inf_{t \in poly}  \limsup_{\nifty} \frac{K_t(\S \restr n)}{n}.
	\end{align*}
\end{definition}

From the results we prove in the following sections, it is implied that if one-way functions exist, then $\Kpoly$ and $\Kstrpoly$ and not equivalent to $\cdimp$ and $\cDimp$.

\subsection{Measure over the Cantor Space}

\begin{definition}[Measure over the Cantor Space \cite{Downey10,Nies2009}]
	Let \( \Sigma^\infty \) be the Cantor space. The  measure of a cylinder set $C_w$ determined by a finite string \( w \in \Sigma^n \) is defined as:
	
	\[
	\mu(C_w) = 2^{-n}.
	\]

\end{definition}

Recall that \( C_w \) denotes the set of infinite sequences in \( \Sigma^\infty \) that begin with \( w \). Using routine measure theoretic arguments, $\mu$ uniquely extends to a probability measure over $(\Sigma^\infty,\mathcal{B}(\Sigma^\infty))$ (see \cite{Billingsley95}). In the rest of the paper we work in the measure space $(\Sigma^\infty,\mathcal{B}(\Sigma^\infty),\mu)$. 

\subsection{One-way functions}

The following is the definition of one-way functions secure against uniform probabilistic polynomial time adversaries.

\begin{definition} [One-way functions \cite{Goldreich2008,vadhan2012,GoldreichCryptography}]
	A function $f_n : \Sigma^* \to \Sigma^*$ is a one-way function (OWF) if:
	\begin{enumerate}
		\item There is a constant $b$ such that $f_n$ is computable in time $n^b$ for sufficiently large $n$.
		\item  For every probabilistic polynomial time algorithm $A$ and every constant $c>0$:
		\[
		\pr[A(f_n(U_n)) \in f_n^{-1}(f_n(U_n))] \leq \frac{1}{n^c}
		\]
		for all sufficiently large $n$, where the probability is taken over the choice of $U_n$ and the internal randomness of algorithm $A$.
	\end{enumerate}
\end{definition}

\subsection{Pseudorandom generators}

In order to define pseudorandom generators, we first define computational indisitinguishability.

\begin{definition}[Computational Indistinguishability \cite{Goldreich2008,vadhan2012}]
	The ensembles of random variables $\{X_n\}_{n \in \N}$ and  $\{Y_n \}_{n \in \N}$ are computationally indistinguishable if for every probabilistic polynomial time algorithm $T$ and every constant $c>0$,
	\begin{align*}
		\left\lvert \pr[T(X_n) = 1] - \pr[T(Y_n) = 1] \right\rvert \leq \frac{1}{n^c}
	\end{align*}
	where the probabilities are taken over the distributions of either $X_n$ or $Y_n$ and the internal randomness of the algorithm $T$.
\end{definition}

\begin{definition} [Pseudorandom generators \cite{Goldreich2008,vadhan2012,GoldreichCryptography}]
	A sequence of functions $G_m : \Sigma^{d(m)} \to \Sigma^m$ is a pseudorandom generator (PRG) if:
	\begin{enumerate}
		\item $d(m) < m$
		\item There exist an uniform and deterministic polynomial time algorithm $M$ such that $M(m,x)=G_m(x)$.
		\item $\{G(U_{d(n)})\}_{n \in \N}$ and $\{U_n\}_{n \in \N}$ are computationally indistinguishable.
	\end{enumerate}
\end{definition}

\subsection{One-way functions and Pseudorandom generators}

The equivalence theorem stated below demonstrates that one-way functions and pseudorandom generators are fundamentally equivalent.

\begin{theorem}[Equivalence of One-way Functions and PRGs \cite{HILL99,Goldreich2008,GoldreichCryptography}] \label{thm:OWFandPRG}
	The following are equivalent:
	\begin{enumerate}
		\item One-way functions exist.
		\item There exist a pseudorandom generator with seed length $d(m) = m-1$.
		\item For every polynomial-time computable constant $\epsilon \in (0,1)$, there exists a pseudorandom generator with seed length $d(m) = \epsilon.m$.
	\end{enumerate}
\end{theorem}

\subsection{Polynomial time samplable distributions on $\Sigma^\infty$}
\label{subsec:ptimesamplabledefinition}

\begin{definition} [ \cite{Royd88}]
	A function  \( \mu: \mathcal{M} \to [0, \infty] \), where $\mathcal{M} \subseteq P(\Sigma
	^\infty)$, is called a measure over $(\Sigma^\infty, \mathcal{M})$ if 
    \( \mu(\emptyset) = 0 \) and $\mu$ is countably additive. That is, for any countable disjoint collection \( \{E_k\}_{k=1}^\infty \) of measurable sets,  
	\[
	\mu\left(\bigcup_{k=1}^\infty E_k\right) = \sum_{k=1}^\infty \mu(E_k).
	\]
	
	We take $\mathcal{M}$ as the set of Borel Measurable sets.
	Note that $ \mathcal{B} \subseteq \mathcal{M} $, where $\mathcal{B}$ is the Borel sigma algebra over $\Sigma^\infty$, we omit $\mathcal{M}$ and say that $\mu$ is a measure over $\Sigma^\infty$. 
	
\end{definition}

Recall that for $w \in \Sigma^*$, $C_w$ denotes the \emph{cylinder set} of $w$, that is the set of all infinite sequences in \( \Sigma^\infty \) that begin with \( w \).
Note that $\mathcal{B}$ contains all open sets, and therefore all cylinder sets $C_w$ for all $w \in \Sigma^*$.

\begin{definition}
	Let $\nu$ be a measure over $\Sigma^\infty$.
	We say that $\nu$ is a probability distribution over $\Sigma^\infty$ if it  holds that  $\nu(\Sigma^\infty) = 1$. 
\end{definition}

\begin{definition} Let $\nu$ be a probability distribution over $\Sigma^\infty$. 
For every $n$, let $\nu_n$ denote the probability distribution induced by $\nu$ on $\Sigma^n$ defined as $\nu_n(w) = \nu(C_w)$ for every $w \in \Sigma^\infty$.
\end{definition}

Now we define polynomial time samplable distributions on $\Sigma^\infty$.

\ptimsampldistn*

%
%

One-way functions are secure against inversion by probabilistic polynomial time adversaries. However, polynomial time gales are defined in terms of computation using deterministic machines. In order to bridge this gap, we define martingales and $s$-gales that are approximable using probabilistic polynomial-time machines. 

 The support of a distribution is the set of points at which the distribution does not vanish. That is $\mathrm{supp}(\nu_n) = \{w \in \Sigma^n : \nu_n(w) > 0\}$.
 
 \ptimapproxdistn*

In other words, the set of strings at which the algorithm $M$ makes an error in approximating $d(w)$ outside $[c.d(w), d(w)]$ lies within the support of $\nu_n$, and the measure of the set of such strings according to $\nu_n$ is inverse polynomial.

\section{One-way functions and polynomial time samplable distributions}	
\label{sec:maintheoremintro}

We show that the existence of dimension gaps almost everywhere with respect to certain polynomial time samplable distributions over $\Sigma^\infty$ is equivalent to the existence of one-way functions. Towards this end, we consider polynomial time samplable distributions $\nu$ that operate on short seeds. Hence, the number of random bits used by the distribution to sample a string of length $n \in \N$ is at most $s.n$ where $s < 1$. This naturally creates a candidate space of strings having low descriptive complexity. For any $X \in range(\nu)$, we have that $\Kpoly(X) \leq s$. 

In this work, we define the notion of a ``significant" gap between Kolmogorov complexity ($\Kpoly)$ and gale based ($\cdimp$) formulations of polynomial-time dimension for polynomial time samplable distributions. The gap is said to be significant if $\cdimp$ and $\Kpoly$ differ for a large collection of the sequences in the range of the distribution. 
 We say that there is no ``significant" dimension gap induced by $\nu$ if for any $s' < s$, an $s'$-gale $d$ would win on a significant fraction of the sequences in $\F$, ensuring that $\nu(S^\infty(d)) > 0$. 
 
We first show that one-way functions imply the existence of distributions over which dimension gap exists almost everywhere.

\equivalenceforward*

We then show the result in the converse direction that the existence of distributions over which dimension gap exists almost everywhere implies the existence of infinitely often one-way functions.

\equivalenceconverse*

Together we obtain the following theorem between dimension gaps and existence of one-way functions.

\equivalencetheorem*

\textbf{Remark on short seed distributions and dimension gaps:} We note that in order to observe dimension gaps over polynomial time samplable distributions, it is essential to restrict the class of distributions to those which uses at most $sn$ random bits to sample a string of length $n$. If we consider the class of distributions which uses more than $n$ bits to sample $n$-length strings, then the uniform distribution $\mu$ is trivially included in this class. In this scenario, the second condition in Lemma \ref{lem:equivalenceconverse} is vacuously true as a consequence of the Kolmogorov inequality for martingales (See Lemma \ref{lem:KolmogorvInequality}). The fact that the sampler for the distribution in the hypothesis of Lemma \ref{lem:equivalenceconverse} uses $sn$-length seeds to sample strings of length $n$, is crucial in our proof of Lemma \ref{lem:equivalenceconverse}.

\section{A mapping between infinite sequences using PRGs}
\label{sec:constructionofg}
One of the major constructions we require in proving our results is an extension of pseudorandom generators with constant stretch to a mapping between infinite strings.

\subsection{Definition of $g$}
 Given a PRG $G_n: \Sigma^{sn} \to \Sigma^n$ where $s = 2^{-m}$ for some $m \in \N$, we construct a mapping $g$ from the Cantor set to itself, by using the family $\{G_n\}_{n \in \N}$ to map blocks within the input string to \emph{longer} blocks within the output string. Given an input $X \in \Sigma^\infty$,  $g(X)$ is constructed by applying $G_n$ to successive block $x$ of size $s.2^n$ from $X$ for all $n \in \N$. We first define a sequence of mappings $g_k$ from $\Sigma^k$ to $\Sigma^*$. These mappings preserves the prefix ordering among finite strings and hence we define $g(X)$ for $X \in \Sigma^\infty$ as the limit of $g_k(X \restr k)$. 
 
\begin{restatable}{definition}{gConstrn} \label{def:gConstrn}
	Let $G_n: \Sigma^{sn} \to \Sigma^n$ where $s = 2^{-m}$ for some $m \in \N$ and let $k \in \N$. For $w \in \Sigma^k$, define $g_k(w)$ as follows. For any $n$ satisfying any $m < n \leq \lfloor \log(k /s) \rfloor$ and $i \in \{2^{n-1}+1,2^{n-1}+2,\dots 2^n\}$, let
	\begin{align*}
	g_k(w)[i] = G_{2^{n-1}}(x)[i - 2^{n-1}]
	\end{align*}
	where $x = w[s \cdot 2^{n-1}, s \cdot 2^n-1]$. For $i \in \{1,2,\dots,2^m\}$, we set $g_k(w)[i] = 0$.
\end{restatable}

Note that $g_k(w)[2^{n-1}+1, 2^{n}] = G_{2^{n-1}}(w[s.2^{{n-1}}, s.2^{n}-1])$. 
 For $k \leq 2^m$, $g_k$ maps any $w \in \Sigma^k$ to $0^{2^m}$. Observe that the sequence of mappings $\langle g_k \rangle_{k \in \N}$ preserves the prefix order between strings. That is, If $w' \sqsubseteq w$, then $g_{\lvert w' \rvert}(w') \sqsubseteq g_{\lvert w \rvert}(w)$. Therefore, we define $g$ as the limit of the mappings $g_k$ as $k \to \infty$.
\begin{definition}
	For every $X \in \Sigma^*$ let $g(X) = \cap_{n=1}^\infty C_{g(X \restr n)}$.
\end{definition}

 Notice that for any $X \in \Sigma^\infty$ and $k \in \N$, $g_k(X \restr k) \sqsubseteq g(X)$. For convenience, we abuse the notation and use $g$ to denote the mapping between infinite strings defined above or a mapping between finite strings defined as follows: $g(w)=g_{\lvert w \rvert}(w)$ for any $w \in \Sigma^*$. The intended meaning is clear from the context. 


\subsection{Analysis of $K_t(g(X) \restr 2^n)$}

Now, we prove the following lemma that the polynomial time bounded Kolmogorov complexity of prefixes of sequences in the image set of $g$ is low.

\begin{lemma}\label{lem:Allender2}
	For $s=2^{-m}$ for some $m > 1$, let $g : \Sigma^\infty \to \Sigma^\infty$ be constructed from a PRG $G_n:\Sigma^{sn} \to \Sigma^n$ that runs in time $t_G(n)$ as in Definition \ref{def:gConstrn}. There exists a constant $c \in \N$ such that for all $X \in \Sigma^*$ and for all $n \in \N$,
	\begin{align*}
		K_{t_g}(g(X)\restr2^n) \leq s.2^n + c
	\end{align*}
	where $t_g(n)=t_G(n)n$.
\end{lemma}
\begin{proof}
Consider the algorithm $M$ which implements the mapping in  Definition \ref{def:gConstrn} as follows.
\begin{algorithm}[H]
	\caption{Algorithm M}
	\begin{algorithmic}[1]
		\State \textbf{Input:} A string $w$.
		\If{$w=\lambda$}
		\State Output $\lambda$
		\Else
		\State Output $0^{2^m}$
		\EndIf
		\For{$i = m$ to $\lfloor \log(\lvert w \rvert 2^m)\rfloor-1$}
		\State Output $G_{2^i}(w[2^{-m}\cdot 2^{i}, 2^{-m} \cdot 2^{i+1}-1])$
		\EndFor
	\end{algorithmic}
\end{algorithm}
Since $G_n$ runs in $t_G(n)$ time, it follows that $M$ runs in $t_g(n)=t_G(n)\cdot n$ time. Also, for any $X \in \Sigma^\infty$ and $n \geq m$, $M$ outputs the string $g(X) \restr 2^n$ on input $X \restr 2^{-m} \cdot 2^n = X \restr s \cdot 2^n$. Therefore, it follows that $K_{t_g}(g(X) \restr 2^n) \leq s\cdot 2^n +c$ where the constant $c$ depends only on the length of any representation of $M$ in the prefix language of the fixed universal Turing machine.
\end{proof}

\begin{corollary}
\label{cor:StrongAllender2}
	For $s=2^{-m}$ for some $m > 1$, let $g : \Sigma^\infty \to \Sigma^\infty$ be constructed from a PRG $G_n:\Sigma^{sn} \to \Sigma^n$ that runs in time $t_G(n)$ as in Definition \ref{def:gConstrn}. There exists a constant $c \in \N$ such that for all $X \in \Sigma^*$ and for all $i \in \N$,
	\begin{align*}
		K_{t_g}(g(X)\restr i) \leq 2s \cdot i + 2\log(i) + c'
	\end{align*}
	where $t_g(n)=t_G(n)n$.
\end{corollary}
\begin{proof}
	From Lemma \ref{lem:Allender2}, we know that
	\begin{align*}
	K_{t_g}(g(X)\restr 2^{\lceil log(i) \rceil}) \leq s  \cdot 2^{\lceil log(i) \rceil} + c'.	
	\end{align*}
	By providing $K_n(i)$ extra information specifying the value of the index $i$ at which $g(X)\restr 2^{\lceil log(i) \rceil}$ needs to be trimmed, we obtain a program that outputs $g(X)\restr i$ in $t_g(n)$-time. There is a trivial prefix free program which given $2\log(i)$ bits of information, runs in linear time and prints $i$. Therefore, the value of $K_n(i)$ is at most $2\log(i)$ up to an additive constant. Since $2^{\lceil log(i) \rceil} \leq 2 i$, there exists $c'>0$ such that
	\begin{align*}
		K_{t_g}(g(X)\restr i) \leq 2s \cdot i + K_n(i) + c \leq 2s \cdot i + 2\log(i) + c'.
	\end{align*}
\end{proof}

\section{Dimension Gap from one-way functions} \label{sec:forwardimplication}

In this section, we prove Lemma \ref{lem:equivalenceforward} by showing that if one-way functions exists, there exists a short polynomial time samplable distribution $\nu$ with a dimension gap, that is for any $s' > s$, and any polytime $\nu$-approximable gale $d$, $\nu(S^\infty(d)) = 0$.

 We start with the assumption that one-way functions exist. From Theorem [], we have that for $s=2^{-m}$ with $m\geq 2$, pseudorandom generators
  $\{G_n : \Sigma^{sn} \to \Sigma^{n}\}$ exist. Note that the mapping $g:\Sigma^\infty \to \Sigma^\infty$ constructed using the ensemble $\{G_n\}$ naturally generates a polynomial-time samplable distribution $\nu$ over $\Sigma^\infty$. The sampling algorithm $M$ for $\nu$ on input $1^n$, maps a random seed of length $s\cdot 2^{\lceil \log_2(n) \rceil}$ using $g$ to a string of length $2^{\lceil \log_2(n) \rceil}$, and thereafter trims the output to $n$ bits. The number of random bits used by $M$ on input $1^n$ is at most $2s \cdot n$. We show that the distribution $\nu$ defined by $g$ satisfies the required property.

Towards building a contradiction, let there exist some $s' \in (2s,1/2)$ with $2^{s'}\in \Q$ and a polynomial-time $\nu$-approximable $s'$-supergale $d'$ such that $\nu(S^\infty(d'))>0$. 
We now convert the polynomial time $\nu$-approximable $s'$-supergale $d'$ to a polynomial-time $\nu$-approximable martingale $\tilde{d}$ with a guarantee that its capital grows \emph{significantly} between blocks of length $2^n$ for infinitely many $n$ and $\nu(S^\infty(\tilde{d}))>0$.

\begin{lemma}
\label{lem:galetofastgrowingmgale}
For $s' \in [0,\infty)$ if $d': \Sigma^\infty \to \Rplus$ is a $t(n)$-time computable $\nu$-approximable $s'$-gale, then for any $s''$ satisfying $s' \leq s'' \leq 1$ and $2^{s''}\in \Q$, there exists an exact $t(n)poly(n)$-time $\nu$-approximable supermartingale $\tilde{d}: \Sigma^\infty \to \Rplus$ such that for any $X \in S^{\infty}(d')$, 
\begin{align*}
	\limsup\limits_{n \to \infty }\frac{\tilde{d}(X \restr n)}{2^{(1-s'')n}} = \infty.
\end{align*}
\end{lemma}
\begin{proof}
Let $s'' \geq s'$ such that $2^{s''} \in \Q$. We know that $d'$ is an exact $t(n)$-time $s''$-supergale that succeeds on $X$ (see Observation 4.4 from \cite{LutzDimension2003}). Let $\bar{d}$ be defined as $\bar{d}(w)=d'(w)2^{(1-s'')\lvert w \rvert}$. Note that as $d'$ is $\nu$-approximable, it follows that $\bar{d}$ is $\nu$-approximable.

From Observation 3.2 in \cite{diss}, it follows that $\bar{d}$ is a supermartingale such that,
\begin{align*}
\limsup\limits_{n \to \infty} \frac{\bar{d}(X \restr n)}{2^{(1-s'')n}} =\infty.	
\end{align*}
Since $2^{s''}$ is rational, so is $2^{(1-s'')}$. Hence, $\bar{d}$ is an exact $t(n)poly(n)$-time computable supermartingale. 
\end{proof}

\begin{lemma}
\label{lem:optimalmartingaleconstruction}
For $s' < 1/2$, let $d': \Sigma^\infty \to \Rplus$ be an  $t(n)$-time $\nu$-approximable $s'$-gale. Let $\tilde{s}$  be such that $\tilde{s}>s'$, $2^{\tilde{s}} \in \Q$. Then, there exists an  $t(n)poly(n)$-time $\nu$-approximable martingale $\tilde{d}: \Sigma^\infty \to \Rplus$ such that for any $X \in S^{\infty}(d')$, there exists infinitely many $n \in \N$ such that 
\begin{align} \label{eq:mgalegrowth2}
	\tilde{d}(X \restr 2^n) > 2^{(1-\tilde{s})2^{n-1}} \tilde{d}(X \restr 2^{n-1})
\end{align}
\end{lemma}
\begin{proof}
Choose $s''\in (s',1/2)$ such that $s'<2s'<2s''<\tilde{s}$ and $2^{s''}\in \Q$. Let $\tilde{d}$ be the exact $t(n)poly(n)$-time computable martingale from Lemma \ref{lem:galetofastgrowingmgale} such that,
\begin{align*}
	\limsup\limits_{n \to \infty }\frac{\tilde{d}(X \restr n)}{2^{(1-s'')n}} = \infty.
\end{align*}
We show that $\tilde{d}$ satisfies the required property. On the contrary, assume that there exists $m>0$ such that for every $n>m$,
\begin{align*} 
	\tilde{d}(X \restr 2^n) \leq 2^{(1-\tilde{s})2^{n-1}} \tilde{d}(X \restr 2^{n-1}).
\end{align*}
Consider any $n > m$. Let $\ell \leq 2^n$. We get that,
\begin{align*}
\tilde{d}(X \restr 2^n+\ell) &\leq 2^\ell 2^{(1-\tilde{s})\sum_{i=m}^{n-1}2^i}	\tilde{d}(X \restr 2^m)\\
&=2^\ell 2^{(1-\tilde{s})(2^n-2^m)} \tilde{d}(X \restr 2^m).
\end{align*}
If $\tilde{d}(X \restr 2^n+\ell) \geq 2^{(2^n+\ell)(1-s'')}$, then
\begin{align*}
2^{\ell s''} \geq \frac{2^{(1-\tilde{s})2^m}}{\tilde{d}(X \restr 2^m)} 2^{2^n(\tilde{s}-s'')}.	
\end{align*}
Since $\ell \leq 2^n$, we obtain
\begin{align}
\label{eq:2nupperbound}
2^{2^n (\tilde{s}-2s'')} \leq \frac{\tilde{d}(X \restr 2^m)}{2^{(1-\tilde{s})2^m}}.	
\end{align}
If there exist infinitely many $k$ such that $\tilde{d}(X \restr k) \geq 2^{(1-s'')k}$ then there exist infinitely many $n$ satisfying \ref{eq:2nupperbound}. But, since $\tilde{s}>2s''$, \ref{eq:2nupperbound} cannot be true for large enough $n$ and therefore we obtain a contradiction. Hence, it must be the case that for infinitely many $n$,
\begin{align*} 
	\tilde{d}(X \restr 2^n) > 2^{(1-\tilde{s})2^{n-1}} \tilde{d}(X \restr 2^{n-1})
\end{align*}
which is the required conclusion.
\end{proof}

\subsection{A Distinguisher Algorithm for PRGs} \label{sec:AlgmA}

Let $\tilde{d}$ be the $t(n).poly(n)$-time $\nu$-approximable martingale obtained from Lemma \ref{lem:optimalmartingaleconstruction}. Let $M$ be the Turing machine that $\nu$-approximates $d$. That is for all $n \in \N$, $\nu_n\{w \in \Sigma^n :M(w) \not\in [\constc \cdot d(w),d(w)]\} \leq n^{-k}$ where $\constc \in \Q,k \in \N$ are constants such that $\constc \leq 1$ and $k \geq 1$. We use $M$ to build a distinguisher algorithm $A$ that breaks the PRG $\{G_n\}$.

\begin{algorithm}[H]
	\caption{Algorithm A}
	\label{alg:distinguisher}
	\begin{algorithmic}[1]
		\State \textbf{Input:} A string $w$ where $|w| = 2^n$
		\State $r \gets \text{random}(s \cdot 2^n)$ \Comment{Generate a random seed of length $s \cdot 2^n$}
		\For{$i = m$ to $n-1$}
		\State $v_i \gets G_{2^i}(r[s \cdot 2^i, s \cdot 2^{i+1}-1])$ \Comment{Use $r$ and $G_{2^i}$ to get a PRG output of length $2^i$}
		\EndFor
		\State $w' \gets 0^{2^m}v_1 \dots v_n$ \Comment{Concatenate all $v_i$}
		\State \Return $1$ \textbf{if} $M(w'w) \geq \constc \cdot 2^{(1 - \tilde{s})\lvert w \rvert} \cdot M(w')$. 
	\end{algorithmic}
\end{algorithm}

\subsection{Analysis of the Distinguisher Algorithm} \label{sec:AlgmAanalys}

\subsubsection{Performance of the algorithm on outputs of PRGs}

Now, we define the function $f$.

\begin{restatable}{lemma}{BorelCant} \label{lem:BorelCant}
	Let $f: \Sigma^\infty \times \N \to \{0,1\}$ and $\mathcal{S} \subseteq \Sigma^\infty$ such that $\nu(\mathcal{S}) > 0$. 
	If for every $X \in \mathcal{S}$, there exist infinitely many $n \in \N$ such that $f(X , n)=1$. Then for any $c>1$, there exist infinitely many $n \in \N$, such that $\nu(\{X: f(X , n)=1\})> n^{-c}$.
\end{restatable}
\begin{proof}
Let $A_n=\{X \in \Sigma^\infty : f(X , n)=1\}$. We know that,
\begin{align*}
	\mathcal{S} \subseteq \limsup\limits_{n \to \infty} A_n = \bigcap\limits_{i=0}^{\infty} \bigcup\limits_{n=i}^{\infty} A_n. 
\end{align*}
Therefore, $\nu(\limsup_{n \to \infty} A_n) \geq \nu(\mathcal{S}) >0$. Since for $c>1$, the series $\sum n^{-c}$ is convergent, using the Borel-Cantelli Lemma (see \cite{Billingsley95}), we obtain that $\nu(A_n) > n^{-c}$ for infinitely many $n$, which is the required conclusion. 	
\end{proof}

\begin{definition}\label{def:f}
Let $f : \Sigma^\infty\times \N \to \{0,1\}$ be defined as follows,
\label{def:fdefinition1}
\begin{align*}
    f(X,n) =
    \begin{cases*}
      1 & if $\tilde{d}(X \restr 2^n  ) > 2^{(1-\tilde{s})2^n}\tilde{d}(X \restr 2^{n-1})$ \\
      0        & otherwise.
    \end{cases*}
  \end{align*}
\end{definition}

\begin{restatable}{lemma}{BorelCantellionsgale} \label{lem:BorelCantellionsgale}%
Given $\tilde{d}: \Sigma^\infty \to \Rplus$, if for all sequences $X \in \mathcal{S} $ where $\nu(\mathcal{S} ) >0$, there exists infinitely many $n \in \N$ such that 
\begin{align*} 
	\tilde{d}(X \restr 2^n) > 2^{(1-\tilde{s})2^{n-1}} \tilde{d}(X \restr 2^{n-1})  
\end{align*}
Then it holds that for infinitely many $n$,
\begin{align} \label{eq:Cond_final}
	\pr_{x \sim \Sigma^{s.2^n}}  [\tilde{d}(g(x )) > 2^{(1-\tilde{s})2^n}\tilde{d}(g(x  \restr s.2^{n-1}))] \geq \frac{1}{n^2}.
\end{align}
\end{restatable}

\begin{proof}

We have that for all $Y \in \mathcal{S} $ where $\nu(\mathcal{S} ) >0$, there exists infinitely many $n \in \N$ such that $
	\tilde{d}(Y \restr 2^n) > 2^{(1-\tilde{s})2^{n-1}} \tilde{d}(Y \restr 2^{n-1}).$
From the definition of the function $f$ (Definition \ref{def:f}), we have that for all $X \in \mathcal{S}$, there exist infinitely many $n \in \N$ such that $f(X , n)=1$. 

Since $\nu(\mathcal{S}) >0$, using Lemma \ref{lem:BorelCant}, we obtain that for any $c>1$ there exist infinitely many $n \in \N$ such that $\nu(\{X: f(X , n)=1\})> n^{-c}$. Then, from the definition of $f$, we have $$\nu(\{X: \tilde{d}(X \restr 2^n  ) > 2^{(1-\tilde{s})2^n}\tilde{d}(X \restr 2^{n-1}) \} > n^{-c}.$$ 

Since $\nu$ is defined in terms of $g$, it follows that $$\pr_{x \sim \Sigma^{s.2^n}}  [\tilde{d}(g(x )) > 2^{(1-\tilde{s})2^n}\tilde{d}(g(x  \restr s.2^{n-1}))] > n^{-c}.$$

Putting $c = 2$ , we have

$$\pr_{x \sim \Sigma^{s.2^n}}  [\tilde{d}(g(x )) > 2^{(1-\tilde{s})2^n}\tilde{d}(g(x  \restr s.2^{n-1}))] > n^{-2}.$$

\end{proof}

%

Let $n+1$ be one of the lengths at which the condition in Lemma \ref{lem:BorelCantellionsgale}  holds. 
Let $w,w' \in \Sigma^{2^n}$ such that $w = G_{2^n}(x)$ for some $x \in \Sigma^{2^{sn}}$ and $w' = g(r)$ for some $r \in \Sigma^{2^{sn}}$. 

In order for the condition $M(w'w) \geq \constc \cdot 2^{(1 - \tilde{s})\lvert w \rvert} \cdot M(w')$ to hold, it suffices that $\tilde{d}(w'w) \geq 2^{(1 - \tilde{s})\lvert w \rvert} \cdot \tilde{d}(w')$ and $M(w'w) \geq \constc \cdot \tilde{d}(w'w)$ and $M(w) \leq \tilde{d}(w)$. 


From Lemma \ref{lem:BorelCantellionsgale}, it follows that the first condition holds with probability at least $(n+1)^{-2}$ among $x,r$  selected uniformly at random from $\Sigma^{s \cdot 2^n} \times \Sigma^{s \cdot 2^n}$ and taking $w = g(x), w' = g(r)$.
Since $\tilde{d}$ is $\nu$-approximable, for some $k \in \N$, on input $w \in \Sigma^{2^n}$,  the computations of all $M(w)$ in Algorithm \ref{alg:distinguisher} outputs a value in $[\constc \cdot \tilde{d}(w), \tilde{d}(w)]$ with probability at least $1-2^{-(kn-1)}$ according to the distribution $\nu_{2^n}$. Note that the distribution according to $\nu_{2^n}$ corresponds to $g$ applied to the uniform distribution on $2^{sn}$. 

Therefore it follows that for infinitely many n, 
\begin{align*}
\pr_{(x,r) \sim \Sigma^{s \cdot 2^n} \times \Sigma^{s \cdot 2^n}}[A(G_{2^n}(x),r)) = 1] &\geq \frac{1}{(n+1)^2}-\frac{1}{2^{kn}}.
\end{align*}

Hence, taking $N=2^n$, for infinitely many $n$,  
 \begin{align}
\pr_{x \sim \Sigma^{s \cdot N}}[A(G_{N}(x)) = 1] \geq \frac{1}{(\log N + 1)^2}-\frac{1}{N^k}. 
\end{align}

\subsubsection{Performance of the algorithm on uniformly random inputs}

We now analyze $\pr_{x \in \Sigma^{2^n}}[A(x) = 1]$. We use the Kolmogorov inequality for supermartingales to show that this is inverse exponentially small in input size.

\begin{lemma} [Kolmogorov Inequality \cite{LutzDimension2003}]\label{lem:KolmogorvInequality}
	Let $d : \Sigma^\infty \to \Rplus$ be a supermartingale. For any $w' \in \Sigma^*$ and $n,c \in \N$, the number of strings $w \in \Sigma^n$ such that $d(w'y) \geq c.d(w')$ for some $y \sqsubseteq w$ is less than or equal to $2^{n}/c$. 
\end{lemma}
\begin{proof}
Define, $$\mathcal{E}=\{x : (\exists w \in \Sigma^n (x \sqsubseteq w )) \land (d(w'x) \geq c \cdot d(w')) \land (\forall y \sqsubset x (d(w'y) < c \cdot d(w')))\}$$ 
Observe that $\mathcal{E}$ is a prefix free set such that for every $w \in \Sigma^n$ satisfying $d(w'y) \geq c.d(w')$ for some $y \sqsubseteq w$, there exists an $x \in \mathcal{E}$ such that $x \sqsubseteq y$. 
Since $d$ is a supermartingale, using an induction argument, it follows that
\begin{align*}
d(w') \geq  \sum\limits_{x \in \mathcal{E}} \frac{1}{2^{ \lvert x \rvert}} \cdot d(w'x) =  \sum\limits_{x \in \mathcal{E}} \frac{2^{n - \lvert x \rvert}}{2^{n}} \cdot d(w'x) \geq \frac{1}{2^{n}} \sum\limits_{x \in \mathcal{E}}2^{n - \lvert x \rvert}d(w'x).
\end{align*}
If $d(w')=0$, then the lemma follows trivially. Assume that $d(w')>0$. Let $M$ be the number of strings $w \in \Sigma^n$ such that $d(w'y) \geq c.d(w')$ for some $y \sqsubseteq w$. It follows that,
\begin{align*}
d(w') \geq  \frac{1}{2^{n}} \cdot M \cdot c \cdot d(w').
\end{align*}
Therefore, we obtain that $M \leq 2^{n}/c$, which is the required conclusion.
\end{proof}



 We need an upper bound on the number of strings $w \in \Sigma^{2^n}$ and $r \in \Sigma^{2^{sn}}$ such that $A(w,r) = 1$. This happens when $M(w'w) \geq \constc \cdot 2^{(1 - \tilde{s})\lvert w \rvert} \cdot M(w')$.  

 For any $n \in \N$, take any $r \in \Sigma^{s.2^n}$, and let $w' = g(r)$. 
 Using the Kolmogorov Inequality for Martingales (Lemma \ref{lem:KolmogorvInequality}), the number of strings $w \in \Sigma^{2^{n}}$ such that $\tilde{d}(w'w) \geq \constc^2 \cdot 2^{|w|(1-\tilde{s})} \tilde{d}(w')$ is less than $2^{|w|} / (\constc^2 \cdot 2^{|w| ( 1- \tilde{s})})$. 
 
 For the remaining set of strings, we have that $\tilde{d}(w'w) < \constc^2 \cdot 2^{|w|(1-\tilde{s})} \tilde{d}(w')$. If we have that $M(w'w) \leq \tilde{d}(w'w)$ and $M(w') \geq \constc \cdot \tilde{d}(w')$, it follows that $A(w,r) = 0$.

 We know that, $\nu\{w \in \Sigma^{2^n}:M(w) \not\in [\constc \cdot \tilde{d}(w),\tilde{d}(w)]\} \leq 2^{-nk}$ for some $k \in \N$. 
 From this, it follows that $ \pr_{r \in \Sigma^{s.N}} [M(w') < \constc \cdot \tilde{d}(w')] \leq 2^{-nk}$.
 
 Notice that $\{w \in \Sigma^j:M(w) \not\in [\constc\cdot \tilde{d}(w),\tilde{d}(w)]\} \subseteq \mathrm{supp}(\nu_j)$ for every $j$. From the block-wise definition of the mapping $g$, for any fixed $w' \in g(\Sigma^{2^n})$, we obtain that there are at most $2^{s \cdot 2^n}$ strings in the set $\{w'w\}_{w \in \Sigma^{2^n}}$ contained in $\mathrm{supp}(\nu_{2^{n+1}})$. Therefore, for any fixed $w'$, the number of strings $w \in \Sigma^{2^n}$ such that $M(w'w) \not\in [\constc\cdot \tilde{d}(w'w),\tilde{d}(w'w)]$ is at most $2^{s\cdot 2^n}$.

  Therefore, we have that $$\pr_{x \sim \Sigma^N} \pr_{r \in \Sigma^{s.N}} [A(x,r) = 1] < \constc^{-2} \cdot 2^{- N ( 1- \tilde{s})} +N^{-k}+2^{-(1-s)N}.$$ Therefore, it follows that for any $n \in \N$,

\begin{align} \label{eq:analysisofA2}
	\pr_{x \sim \Sigma^N } [A(x) = 1] < \frac{1}{\constc^2 \cdot 2^{ N  ( 1- \tilde{s})}}+\frac{1}{N^k}+\frac{1}{2^{ N ( 1- {s})}}.
\end{align}

\subsection{Proof of Lemma \ref{lem:equivalenceforward}}

Now, we conclude the proof of Lemma \ref{lem:equivalenceforward}.

\begin{proof}[Proof of Lemma \ref{lem:equivalenceforward}]
If one-way functions exist then there exists PRG $G_n:\Sigma^{sn} \to \Sigma^n$ running in time $t_G(n) \in \poly(n)$ [Theorem \ref{thm:OWFandPRG}]. 
For all $s=2^{-m}$ with $m\geq 2$, using $\{G_n\}$, we construct a polynomial time samplable distribution $\nu$ over $\Sigma^\infty$  such that the number of random bits used by the sampler for $\nu$ on input $1^n$ is at most $sn$. 

We further show that if we assume to the contrary that for some $s' \in (s,1/2)$ and some polynomial-time $\nu$-approximable $s'$-supergale $d$, $\nu(S^\infty(d)) > 0$, then there  exists a distinguisher algorithm $A$ that runs in time $t_A(n) \in \poly(n)$, such that for infinitely many $N \in \N$, for all $k \in \N$,
\begin{align*} 
	\pr_{x \in \Sigma^{s.N}} [A(G_N(x)) = 1] -  &	\pr_{x \in \Sigma^{N}} [A(x) = 1] > \\ & \frac{1}{(\log N + 1)^2}-\frac{2}{N^k} -  \frac{1}{\constc^2 \cdot 2^{ N  ( 1- \tilde{s})}}-\frac{1}{2^{ N ( 1- {s})}}. 
\end{align*}

 Let $c$ be any constant such that $t_A(n) \leq n^c$ and for all large enough $N$,
 \begin{align*}
 	\frac{1}{(\log N + 1)^2}-\frac{2}{N^k} -  \frac{1}{\constc^2 \cdot 2^{ N  ( 1- \tilde{s})}}-\frac{1}{2^{ N ( 1- {s})}} \geq \frac{1}{N^c}
 \end{align*}
It follows that $G_N$ is not a $(N^c, 1/N^c)$ PRG for infinitely many $N$. Therefore the family $\{G_n\}_{n \in \N}$ is not a PRG. Hence we have shown that PRG's $G_n:\Sigma^{sn} \to \Sigma^n$ do not exist for $s=2^{-m}$ where $m\geq 2$, which contradicts the assumption that One-way functions exist. This completes the proof of the lemma.

 \end{proof}

\section{Infinitely often one-way functions from Dimension Gaps}
\label{sec:backwardimplication}

 We require the following result (Theorem 20) from \cite{ilango2022robustness}. 


\begin{theorem}[\cite{ilango2022robustness}]
\label{thm:irsextrapolation}
Assume that no infinitely-often one-way functions exist. Let $\mathcal{D}=\{\mathcal{D}_n\}$ be a polynomial time samplable distribution and $q \geq 1$ be an arbitrary constant. Then, there exists a probabilistic polynomial time algorithm $\mathcal{A}$ and constant $c<1$ such that for all $n$,
\begin{align*}
\pr_{x \sim \mathcal{D}_n} [c \cdot\mathcal{D}_n(x) \leq \mathcal{A}(x) \leq \mathcal{D}_n(x)]	\geq 1-O\left(\frac{1}{n^q}\right).
\end{align*}
\end{theorem}

 Now, we prove Lemma \ref{lem:equivalenceconverse}.

\equivalenceconverse*

\begin{proof}
	Assume that for some $s<1$ there exist a polynomial time samplable distribution $\nu$ over $\Sigma^\infty$ such that the number of random bits used by the sampler for $\nu$ on input $1^n$ is at most $sn$ and there exists $\hat{s} \in (s,1]$ such that for every $s' \in (s,\hat{s})$ and polynomial-time $\nu$-approximable $s'$-supergale $d$, $\nu(S^\infty(d))=0$.

	Fix arbitrary $s'$ such that $s<s'<\hat{s}$ and $2^{s'}\in \Q$. Consider the function $d:\Sigma^* \to [0,\infty)$ defined as $d(w)=2^{s'\lvert w \rvert} \nu(w)$. It is easy to verify that $d$ is an $s'$-supergale. Since the number of random bits used by the sampler for $\nu$ on input $1^n$ is at most $sn$, for any $w \in \supp(\nu_n)$, 
	\begin{align*}
		\nu_{n}(w) \geq \frac{1}{2^{n \cdot s}}.
	\end{align*}
	From the above, we obtain
	\begin{align*}
		d(w) \geq 2^{n (s'-s)}.
	\end{align*}
	Take any $X \in \Sigma^\infty$, such that $X \in \range(\nu)$. That is, for all $n \in \N$,  $X \restr n \in \supp(\nu_n)$. It follows that $\limsup_{n \to \infty} d(X \restr n) = \infty$.  From the continuity of measure\cite{Billingsley95}, we have that $\nu(\range(\nu)) = 1$.

	Therefore, it follows that,
	\begin{align*}
		\nu\left\{X \in \Sigma^\infty: \limsup_{n \to \infty} d(X \restr n) = \infty\right\}=1.	
	\end{align*}

	We now show that if infinitely-often one-way functions do not exist, there exists an algorithm $M$ that $\nu$-approximates $d$. Hence, we obtain a contradiction to our assumption that for every $s' \in (s,\hat{s})$ and any polynomial-time $\nu$-approximable $s'$-supergale $d$, $\nu(S^\infty(d))=0$. Therefore, it must be the case that infinitely-often one-way functions exist. 
	
	Let $S$ be the machine that samples the distribution $\nu$ such that for every $n$, $\pr_{r} \left[ S(1^n,r) = w\right]=\nu_n(w)$ for every $w \in \Sigma^n$. Let $c'$ be the constant such that $M$ uses $n^{c'}$ random bits to sample a string of length $n$.
	
	  Assume that no infinitely-often one-way functions exist.  Then, the function $f:\Sigma^* \to \Sigma^*$ defined as $$f(x)=M(1^{\lvert x \rvert^{1/c'}},x)$$ is not an infinitely-often weak one-way function. Hence for any $q \geq 1$ there exist an algorithm $\mathcal{I}$ such that for all $n$,	\begin{align}
	\label{eq:invertercondition}
\pr_{r \sim \Sigma^{n^{c'}}} [ f(\mathcal{I}(f(r)))=f(r)] = \pr_{w \sim \nu_n} [ f(\mathcal{I}(w))=w]	\geq 1-O\left(\frac{1}{n^q}\right).
\end{align}
	
	From our assumption that infinitely-often one-way functions do not exist and Theorem \ref{thm:irsextrapolation}, we obtain that for any $q \geq 1$, there exists a probabilistic polynomial time algorithm $\mathcal{A}$ and constant $c<1$ such that for all $n$,
\begin{align}
\label{eq:extrapolatorcondition}
\pr_{w \sim \nu_n} [c \cdot\nu_n(w) \leq \mathcal{A}(w) \leq \nu_n(w)]	\geq 1-O\left(\frac{1}{n^q}\right).
\end{align}


 Finally, we give the algorithm $M$ that $\nu$-approximates $d$.
 $M$ on input $w \in \Sigma^n$ runs $\mathcal{I}$ on $w$. If $f(\mathcal{I}(w)) \neq w$, $M$ outputs $0$. Else, $M$ runs $\mathcal{A}$ on $w$ and outputs $2^{s'\lvert w \rvert}\mathcal{A}(w)$. Since $2^{s'}\in\Q$, $M$ runs in polynomial time. For any $w$ such that $\nu(w)=0$ (i.e. $w \not\in \mathrm{supp}(\nu_n)$), $\mathcal{I}$ will never succeed in finding a pre-image $r$ such that $f(r) = w$ and hence $M$ always outputs $0$. Therefore for every $n$, $\{w:M(w) \not\in [c\cdot d(w),d(w)]\} \subseteq \mathrm{supp}(\nu_n)$. 

Now, from (\ref{eq:invertercondition}) and (\ref{eq:extrapolatorcondition}) we obtain that,
\begin{align*}
\pr_{x \sim \nu_n} [c \cdot d(x) \leq M(x) \leq d(x)]	\geq 1-O\left(\frac{1}{n^q}\right)-O\left(\frac{1}{n^q}\right) \geq  1-O\left(\frac{1}{n^q}\right).
\end{align*}
Therefore for every $n$, $\nu\{w:M(w) \not\in [c\cdot d(w),d(w)]\} \leq O\left(n^{-q}\right)$. Since $\mathcal{A}$ runs in randomized polynomial time and $q$ is arbitrary, $M$ witness the fact that $d$ is a polynomial time $\nu$-approximable martingale.
\end{proof}

\section{One-way functions and polynomial-time dimension of sets}
\label{sec:setsconjecture}
We prove the following.

\thmforSetspsrg*

In fact, we prove the following stronger assertion.
\begin{theorem}
\label{thm:owfimpliesdimensiongaponsets}
	If one-way functions exist, then for every $s<1/2$, there exist $\F \subseteq \Sigma^\infty$ such that $\Kstrpoly(\F) \leq s$ and $\cdimp(\F)\geq 1/2$.
\end{theorem}

\begin{proof}
Assume that one-way functions exist. From Lemma \ref{lem:equivalenceforward}, for every $s<1/2$ there exist a polynomial time samplable distribution $\nu$ over $\Sigma^\infty$ and polynomial $t$ such that:
\begin{enumerate}
\item The number of random bits used by the sampler for $\nu$ on input $1^n$ is at most $sn$.
\item For every $s' \in (s,1/2)$ and any polynomial-time $\nu$-approximable $s'$-supergale $d$, $\nu(S^\infty(d))=0$.
\end{enumerate}
Consider the following set,
\begin{align*}
\mathcal{F}'=\bigcap\limits_{n \geq 1} \bigcup\limits_{w \in \mathrm{supp}(\nu_n)} C_w.
\end{align*}
Observe that $\cup_{w \in \mathrm{supp}(\nu_n)} C_w$ is a finite union of cylinder sets in $\Sigma^\infty$ and hence is a closed set. Since $\nu$ is a polynomial time samplable probability distribution on $\Sigma^\infty$, $\cup_{w \in \mathrm{supp}(\nu_n)} C_w$ is a superset of $\cup_{w \in \mathrm{supp}(\nu_{n+1})} C_w$ for every $n$. Since $\mathcal{F}'$ is an intersection of a decreasing sequence of non-empty closed sets in the compact space $\Sigma^\infty$, $\mathcal{F}'$ is non-empty (see \cite{simmons1963introduction}). Since $\nu(\cup_{w \in \mathrm{supp}(\nu_n)} C_w)=1$ for every $n$, using the continuity of probability measures from above \cite{Billingsley95}, we get $\nu(\mathcal{F}')=1$. Since the number of random bits used by the sampler for $\nu$ on input $1^n$ is at most $sn$, $\Kstrpoly(\F')\leq s$. Since the number of random bits used by the sampler for $\nu$ on input $1^n$ is at most $sn$, $\Kstrpoly(\F)\leq s$. For the sake of contradiction, assume that $\cdimp(\F)<1/2$. Therefore, for some $s'\in (\cdimp(\F),1/2)$, there exist a polynomial-time  $s'$-supergale $d$ such that $\F \subseteq S^{\infty}(d)$. Since $d$ is polynomial-time computable, $d$ is trivially a polynomial-time $\nu$-approximable $s'$-supergale. Now, since we know that $\nu(S^\infty(d)) \geq \nu(\F)=1$, we obtain a contradiction. Hence, it must be the case $\cdimp(\F)\geq 1/2$.
\end{proof}

\section{One-way functions and polynomial-time dimension of sequences}
\label{sec:sequenceconjecture}

Towards this, we prove the following lemma.

\begin{lemma}
\label{lem:kstrpolyofdistribution}
If $\nu$ is a polynomial time samplable distribution over $\Sigma^\infty$ such that the number of random bits used by the sampler for $\nu$ on input $1^n$ is at most $sn$, then $\nu\{X \in \Sigma^\infty : \Kstrpoly(X) \leq s\}=1$. 
\end{lemma}

We prove the following.

\thmforSeqpsrg*

It is enough to show the following.
\begin{theorem}
\label{thm:sequencestrongtheorem}
	If one-way functions exist, then for every $s<1/2$, there exist $X \in \Sigma^\infty$ such that $\Kstrpoly(X) \leq s$ and $\cdimp(X)\geq 1/2$.
\end{theorem}

We require the following construction.

\begin{restatable}{lemma}{Mgalecombination}  \label{lem:Mgalecombination}
For any time $t(n) \in \poly(n)$, and $s>0$ with $2^{s}\in \Q$, there exists an exact $ t(n)\cdot n \cdot \log(n)$-time computable $s$-gale $\tilde{d}$ such that for all exact $t(n)$-time computable $s$-gales $d$, there exist a constant $c$ such that for all $X \in \Sigma^\infty$ and $n \in \N$, $\tilde{d}(X \restr n ) \geq c \cdot d(X \restr n)$. 
\end{restatable}

\begin{proof}
	Let $\{M_i\}$ be a standard computable enumeration of all Turing machines. The idea of the proof is that we consider the $t(n)$-time martingales among the first $\log n$ machines to place bets on strings of length $n$. The remaining martingales are assumed to bet evenly on all strings of length $n$. 
	
	We will therefore define the martingale, conditioned on the length of the input string $x$.  When the length of the input string $x$ satisfies $2^n \leq |x| < 2^{n+1}$, we consider the machines $M_1 \dots M_n$ to decide $\tilde{d}(x)$. 
	
	The remaining machines $M_i$ for $i \geq n+1$ are assumed to bet \emph{evenly} on $x$. So, for $i \geq n+1$, we force the capital placed by $M_i$ on $x$ to be $2^{|x| \cdot (s - 1)}$. 

	Therefore, given $x \in \Sigma^*$, such that $2^n \leq |x| < 2^{n+1}$, we define 
	\begin{align*}
		\tilde{d}(x) &= \sum_{i =1}^n 2^{-i} \cdot M_i'(x) + \sum_{i=n+1} ^\infty {2^{-i}} \cdot 2^{|x| \cdot (s - 1)} \\
		&= 2^{- n + |x| \cdot (s - 1)} + \sum_{i =1}^n  2^{-i} \cdot M_i'(x). 
	\end{align*}
	
	We now define $M_i'(x)$. For $n \in N$, let $\D_n(x)$ be the set of machines in $\{M_1, \dots M_n\}$ such that for each $j \leq |x|$, $M(x \restr j)$ runs in time $t(j)$ and $M(\lambda) = 1$ and $2^s \cdot M(x \restr j) = M((x\restr j).0) + M((x\restr j).1)$. 
	
	\textbf{Case 1: $M_i\in \D_n(x)$}. The actual bets placed by martingale $M_i$ is taken into consideration for strings of length more than $2^i$. In that case, we need to ensure that  the capital of $M_i$ at $x \restr 2^i$ is forced to be $2^{2^i \cdot (s - 1)}$.
	
	In this case, when $M_i(x \restr 2^i) > 0$, we define $M_i'(x)= 2^{2^i \cdot (s - 1)} (M_i(x)/M_i(x \restr 2^i))$. If $M_i(x \restr 2^i) = 0$, then it holds that $M_i(x) = 0$, so we just take $M_i'(x) = 0$. 
	
	\textbf{Case 2: $M_i\not\in \D_n(x)$}. In this case, let $j < |x|$ be the least number such that $M_i(x\restr j)$ takes more than $t(j)$ time to run, or the $s$-gale condition gets violated, that is $2^s \cdot M_i(x \restr j) \neq M_i((x\restr j)0) + M_i((x\restr j)1)$. In this case, we freeze $M_i's$ capital at index $j$, and $M_i$ is forced to bet evenly after that. Note that we have to force $M_i$ to bet evenly up to $x \restr 2^i$. Therefore if $j < 2^i$, we take $M_i'(x) = 2^{|x| \cdot (s - 1)}$. Otherwise, $j \geq 2^i$ and  $M_i'(x) = M_i'(x \restr j) \cdot 2^{(|x|-j) \cdot (s - 1)}$. Note that from the previous case, we have $M_i'(x \restr j) = 2^{2^i \cdot (s - 1)} (M_i(x \restr j)/M_i(x \restr 2^i))$.
	
	\medskip
	
	Consider any exact $t$-time computable martingale $d$. Let $M_k$ be any $t$-time machine such that $d(w)=M_k(w)$ for every $w \in \Sigma^*$. Then, it follows from the construction of $\tilde{d}$ that for all $w$ with $\lvert w \rvert > 2^k$,	 \begin{align*}
		\tilde{d}(w) \geq \frac{2^{-k} \cdot 2^{-2^k (1-s)}  }{d(w \restr2^k)} \cdot d(w).
	\end{align*}
	
	Let $c$ be the constant on the right hand side. Note that the membership of any machine from $\{M_1,M_2,\dots M_n\}$ in $\D_n(w)$ is decidable in time $t(|w|) \cdot n \cdot |w|$. Since $n \leq \log(|w|)$, and $2^s \in \Q$, it follows that, $d$ is an exact $t(n) \cdot n \cdot \log(n)$-time computable $s$-gale. 
	%
	
\end{proof}

Now, we prove Theorem \ref{thm:sequencestrongtheorem}.

\begin{proof}[Proof of Theorem \ref{thm:sequencestrongtheorem}]
	Assume that one-way functions exist. From Lemma \ref{lem:equivalenceforward}, for every $s<1/2$ there exist a polynomial time samplable distribution $\nu$ over $\Sigma^\infty$ and polynomial $t$ such that:
\begin{enumerate}
\item The number of random bits used by the sampler for $\nu$ on input $1^n$ is at most $sn$.
\item For every $s' \in (s,1/2)$ and any polynomial-time $\nu$-approximable $s'$-supergale $d$, $\nu(S^\infty(d))=0$.
\end{enumerate}
Consider the following set,
\begin{align*}
\mathcal{F}=\bigcap\limits_{n \geq 1} \bigcup\limits_{w \in \mathrm{supp}(\nu_n)} C_w.
\end{align*}
Observe that $\cup_{w \in \mathrm{supp}(\nu_n)} C_w$ is a finite union of cylinder sets in $\Sigma^\infty$ and hence is a closed set. Since $\nu$ is a polynomial time samplable probability distribution on $\Sigma^\infty$, $\cup_{w \in \mathrm{supp}(\nu_n)} C_w$ is a superset of $\cup_{w \in \mathrm{supp}(\nu_{n+1})} C_w$ for every $n$. Since $\mathcal{F}$ is an intersection of a decreasing sequence of non-empty closed sets in the compact space $\Sigma^\infty$, $\mathcal{F}$ is non-empty (see \cite{simmons1963introduction}). Since $\nu(\cup_{w \in \mathrm{supp}(\nu_n)} C_w)=1$ for every $n$, using the continuity of probability measures from above \cite{Billingsley95}, we get $\nu(\mathcal{F})=1$. Since the number of random bits used by the sampler for $\nu$ on input $1^n$ is at most $sn$, $\Kstrpoly(X)\leq s$ for every $X \in \F$. For the sake of contradiction, assume that there exist $s' \in (s,1/2)$ such that $\cdimp(X)<s'$ for every $X \in \F$. So, every $X \in \F$ there exists an exact computable $s'$-gale $d'_X$ that runs in time $t_{d'_X}(n) \in \poly(n)$, such that $d'_X$ succeeds on $X$.

 We first partition the set $\F$, based on the running time of the $s'$-gales that succeeds on it. 
Define 
\begin{align*}
R_i = \{X \in \F : \text{ there exists an $n^i$-time $s'$-gale $d$ such that $X \in S^{\infty}(d)$} \}
\end{align*}
and let $S_i = R_i \setminus R_{i-1}$. Note that $\{S_i\}_{i \in \N}$ is a partition of $\F$ into countably many disjoint subsets. As $\sum_{i=1}^{\infty} \nu(S_i) = 1$, there exists an $k \in \N$ such that $\sum_{i=1}^{k} \nu(S_i) > 0$. Define $\mathcal{S} = \bigcup_{i=1}^{k} S_i$.
Using Lemma \ref{lem:Mgalecombination}, there exists an exact $n^{k+c}$-time computable $s'$-gale $d'$ that succeeds on all sequences in $\mathcal{S}$. Since $d$ is polynomial-time computable, $d$ is trivially a polynomial-time $\nu$-approximable $s'$-supergale. Now, since we know that $\nu(S^\infty(d)) \geq \nu(\mathcal{S})>0$, we obtain a contradiction. Hence, it must be the case that there exist a sequence $X\in \F$ with $\cdimp(X)>s'$. Since $s'$ was arbitrary, the theorem follows.
\end{proof}

We also note the following lemma which follows using the same ideas in the proof of Theorem \ref{thm:sequencestrongtheorem}.

\begin{lemma}
\label{lem:owfimpliesdimensiongapwithprobability1}
	If for some $s<1/2$, $\nu$ is a polynomial time samplable distribution over $\Sigma^\infty$ such that:
\begin{enumerate}
\item The number of random bits used by the sampler for $\nu$ on input $1^n$ is at most $sn$.
\item For every $s' \in (s,1/2)$ and any polynomial-time $\nu$-approximable $s'$-supergale $d$, $\nu(S^\infty(d))=0$.
\end{enumerate}
Then, $$\nu\{X \in \Sigma^\infty : \Kstrpoly(X) \leq s\text{ and } \cdimp(X) \geq 1/2>s\}=1.$$
\end{lemma}
\begin{proof}
Consider the set $\F$ in the proof of Theorem \ref{thm:sequencestrongtheorem}. It is enough to show that, $\nu\{X \in \Sigma^\infty : \Kstrpoly(X) > s\text{ or } \cdimp(X) < 1/2\}=0$. It follows directly from the argument in the proof of Theorem \ref{thm:sequencestrongtheorem} that $\Kstrpoly(X) \leq s$ for every $X \in \F$. Consider the set $\mathcal{H}=\{X \in \F :  \cdimp(X) < 1/2\}$. If $\mu(\mathcal{H})>0$ then there exist some $\hat{s}<1/2$ such that $\mu(\hat{\mathcal{H}})>0$ where $\hat{\mathcal{H}}=\{X \in \F :  \cdimp(X) < \hat{s}\}$. By using the partitioning argument along with the gale combination technique in the in the proof of Theorem \ref{thm:sequencestrongtheorem} over the set $\hat{\mathcal{H}}$ instead of $\F$, we obtain a contradiction to the second statement in the hypothesis of the lemma. Therefore, it must be the case that $\mu(\mathcal{H})=0$ which completes the proof of the lemma.
\end{proof}

\section{Sets and sequences with maximal gap between $\Kstrpoly$ and $\cdimp$}
\label{sec:maximalgaps}

In this section we prove that if one-way functions exist, the gap between $\cdimp$ and $\Kstrpoly$ can be arbitrarily close to $1$.
\begin{theorem}
\label{thm:strongdimensionandgapcombined}
If one-way functions exist, then for every $\epsilon>0$ there exist set $\F$ and sequence $X$ such that $\cdimp(\F)-\Kstrpoly(\F) \geq 1-\epsilon$ and $\cdimp(X)-\Kstrpoly(X) \geq 1-\epsilon$ respectively.	
\end{theorem}

We prove the existence of a sequence $X$ such that $\cdimp(X)-\Kstrpoly(X) \geq 1-\epsilon$. The assertion for sets follows trivially from this claim. Let $s=2^{-m}$ for some $m>1$. We have $s<1/2$ and therefore $2s < 1$. Consider the set $\F'_{2s} = \{X \in \Sigma^\infty : \limsup_{\nifty} \frac{K_{t_g}(\S \restr n)}{n} \leq 2s\}$. For any $X \in \F'_{2s}$, 
\begin{align*}
\Kstrpoly(X)=\inf_{t \in poly}  \limsup_{\nifty} \frac{K_t(X \restr n)}{n} \leq 2s.	
\end{align*}

We argue that there does exist any $s'<1$ satisfying $2s<s'<1$ and $2^{s'}\in \Q$ such that $\cdimp({X})<s'$ for every $X \in \F'_{2s}$. Assume that for some such $s'$ and every $X \in \F'_{2s}$ there exists an exact computable $s'$-gale $d'_X$ that runs in time $t_{d'_X}(n) \in \poly(n)$, such that $d'_X$ succeeds on $X$. 

 As in the proof of Theorem \ref{thm:forSetspsrg}, we construct the mapping $g:\Sigma^\infty \to \Sigma^\infty$ as in Definition \ref{def:gConstrn}. For any $X \in \Sigma^\infty$, from Corollary \ref{cor:StrongAllender2}, it follows that $g(X) \in \F_{2s'}$. Therefore, $g(\Sigma^\infty) \subseteq \F_{2s'}$. 
 
  So, every $X \in \Sigma^\infty$ there exists an exact computable $s'$-gale $d'_X$ that runs in time $t_{d'_X}(n) \in \poly(n)$, such that $d'_X$ succeeds on $g(X)$. Similar to the proof of Theorem \ref{thm:thmforSeqpsrg}, using a gale combination technique, we construct a set $\mathcal{S} \subseteq \Sigma^\infty$ such that $\mu(\mathcal{S})>0$ and
 a single $s$-gale $d'$ that succeeds on $g(X)$ for all $X \in \mathcal{S}$.
 
 We first partition sequences in the input space of $g$ into sets, based on the running time of the $s'$-gales that succeeds on it. 
Define 
\begin{align*}
R_i = \{X \in \Sigma^\infty : \text{ there exists an $n^i$-time $s'$-gale $d$ such that $g(X) \in S^{\infty}(d)$} \}
\end{align*}
and let $S_i = R_i \setminus R_{i-1}$. Note that $\{S_i\}_{i \in \N}$ is a partition of $\Sigma^\infty$ into countably many disjoint subsets. As $\sum_{i=1}^{\infty} \mu(S_i) = 1$, there exists an $k \in \N$ such that $\sum_{i=1}^{k} \mu(S_i) > 0$. Define $\mathcal{S} = \bigcup_{i=1}^{k} S_i$.
Now, using Lemma \ref{lem:Mgalecombination}, there exists an exact $t(n)$-time computable $s'$-gale $d'$ that succeeds on all sequences in $g(\mathcal{S})$ for some polynomial $t(n)$.

Now, using Lemma \ref{lem:Mgalecombination}, there exists an exact $t(n)$-time computable $s'$-gale $d'$ that succeeds on all sequences in $g(\mathcal{S})$ for some polynomial $t(n)$. The major tool we require is the following generalization of Lemma \ref{lem:optimalmartingaleconstruction} to any $s' < 1$. 

\begin{lemma}
\label{lem:optimalmartingaleconstructionforgenerals}
For $s' < 1$, let $d': \Sigma^\infty \to \Rplus$ be an exact $t(n)$-time computable $s'$-gale. Let $\tilde{s}$ and $s''$ be such that $\tilde{s}>s''>s'$, $2^{\tilde{s}} \in \Q$ and $2^{s''}\in \Q$. Then, there exists an exact $t(n)poly(n)$-time computable martingale $\tilde{d}: \Sigma^\infty \to \Rplus$ such that for any $X \in S^{\infty}(d')$, there exists infinitely many $n \in \N$ satisfying at least one of the following:
\begin{enumerate}
	\item\label{item:optimalmgale1} $\tilde{d}(X \restr 2^n) > 2^{(1-\tilde{s})2^{n-1}} \tilde{d}(X \restr 2^{n-1})$.
	\item\label{item:optimalmgale2} There exists $\ell$ satisfying $ ((\tilde{s}-s'')/s'') \cdot 2^{n-1} \leq \ell \leq 2^{n-1}$ such that $\tilde{d}(X \restr 2^{n-1} + \ell) > 2^{(1-s'')(2^{n-1}+\ell)}$. 
\end{enumerate}
\end{lemma}
\begin{proof}
Let $\tilde{d}$ be the exact $t(n)poly(n)$-time computable martingale from Lemma \ref{lem:galetofastgrowingmgale} such that,
\begin{align*}
	\limsup\limits_{n \to \infty }\frac{\tilde{d}(X \restr n)}{2^{(1-s'')n}} = \infty.
\end{align*}
We show that $\tilde{d}$ satisfies the required property. On the contrary, assume that there exists $m>0$ such that for every $n>m$, $\tilde{d}(X \restr 2^n) \leq 2^{(1-\tilde{s})2^{n-1}} \tilde{d}(X \restr 2^{n-1})$ and $\tilde{d}(X \restr 2^n + \ell) \leq 2^{(1-s'')(2^{n-1}+\ell)}$ for every $\ell$ such that $ ((\tilde{s}-s'')/s'') \cdot 2^{n-1} \leq \ell \leq 2^{n-1}$. Since, there exist infinitely many $k$ such that $\tilde{d}(X \restr k) \geq 2^{(1-s'')k}$, we obtain that infinitely many $n$ and $\ell \leq ((\tilde{s}-s'')/s'') \cdot 2^{n-1}$ should satisfy,
\begin{align}
\label{eq:optimalmartingaleeq1}
\tilde{d}(X \restr 2^{n-1} + \ell) > 2^{(1-s'')(2^{n-1}+\ell)}.
\end{align}
But for any large enough $n$ and $\ell$ satisfying $\ell \leq ((\tilde{s}-s'')/s'') \cdot 2^{n-1}$, we have
\begin{align*}
\tilde{d}(X \restr 2^{n-1}+\ell) &\leq 2^\ell 2^{(1-\tilde{s})\sum_{i=m}^{n-2}2^i}	\tilde{d}(X \restr 2^m)\\
&=2^\ell 2^{(1-\tilde{s})(2^{n-1}-2^m)} \tilde{d}(X \restr 2^m).
\end{align*}
From \ref{eq:optimalmartingaleeq1}, we obtain 
\begin{align*}
2^{\ell s''} \geq \frac{2^{(1-\tilde{s})2^m}}{\tilde{d}(X \restr 2^m)} 2^{2^{n-1}(\tilde{s}-s'')}.	
\end{align*}
Since $\ell \leq ((\tilde{s}-s'')/s'') \cdot 2^n$,
\begin{align}
2^{2^{n-1} (\tilde{s}-s'')} \leq \frac{\tilde{d}(X \restr 2^m)}{2^{(1-\tilde{s})2^m}}.	
\end{align}
Since $\tilde{s} > s''$, the above cannot be true for large enough $n$. So, we obtain a contradiction. Therefore, it must be the case that there exist infinitely many $n$ satisfying conditions \ref{item:optimalmgale1} or \ref{item:optimalmgale2}.
\end{proof}

\begin{definition}
\label{def:fdefinition2}
Let $f' : \Sigma^\infty \times \N \to \{0,1\}$ be defined as follows. $f'(X,n)=1$ if either of the following are true:
\begin{enumerate}
	\item $\tilde{d}(g(X ) \restr 2^n) > 2^{(1-\tilde{s})2^n}\tilde{d}(g(X ) \restr 2^{n-1})$.
	\item There exists $\ell$ satisfying $ ((\tilde{s}-s'')/s'') \cdot 2^{n-1} \leq \ell \leq 2^{n-1}$ such that $\tilde{d}(g(X) \restr 2^{n-1} + \ell) > 2^{(1-s'')(2^{n-1}+\ell)}$.
\end{enumerate}
Otherwise we define $f'(X,n)=0$.
\end{definition}
Now, we prove an analogue of Lemma \ref{lem:BorelCantellionsgale} for the new function $f'$. The proof is identical to that of Lemma \ref{lem:BorelCantellionsgale} and uses the Borel-Cantelli lemma.

\begin{restatable}{lemma}{BorelCantellionsgale2} \label{lem:BorelCantellionsgale2}%
	Given $\tilde{d}: \Sigma^\infty \to \Rplus$, if for all sequences $Y \in g(\mathcal{S} )$ where $\mu(\mathcal{S} ) >0$, there exists infinitely many $n \in \N$ satisfying at least one of the following:
	\begin{enumerate}
		\item\label{item:optimalmgale11} $\tilde{d}(X \restr 2^n) > 2^{(1-\tilde{s})2^{n-1}} \tilde{d}(X \restr 2^{n-1})$.
		\item\label{item:optimalmgale21} There exists $\ell$ satisfying $ ((\tilde{s}-s'')/s'') \cdot 2^{n-1} \leq \ell \leq 2^{n-1}$ such that $\tilde{d}(X \restr 2^{n-1} + \ell) > 2^{(1-s'')(2^{n-1}+\ell)}$. 
	\end{enumerate}
	For any $n > 0$, let $\mathcal{H}_n $ be the set of strings in $ \Sigma^{s.2^n}$ such that $\tilde{d}(g(x )) > 2^{(1-\tilde{s})2^n}\tilde{d}(g(x  \restr s.2^{n-1})) $ or there exists $\ell$ satisfying $ ((\tilde{s}-s'')/s'') \cdot 2^{n-1} \leq \ell \leq 2^{n-1}$ such that $\tilde{d}(x \restr 2^{n-1} + \ell) > 2^{(1-s'')(2^{n-1}+\ell)}$. Then it holds that for infinitely many $n$,
	\begin{align*} 
		\pr_{x \in \Sigma^{s.2^n}} \left[x \in \mathcal{H}_n \right] \geq \frac{1}{n^2}.
	\end{align*}
\end{restatable}

Let $q$ be a rational number such that 
\begin{align*}
0<\frac{(\tilde{s}-s'')}{2s''} \leq q \leq \frac{(\tilde{s}-s'')}{s''}.
\end{align*}
Now we modify Algorithm A (see Algorithm \ref{alg:distinguisher}) such that the last step is,
\begin{align*}
	\textbf{return } \tilde{d}(w'w) \geq 2^{(1 - \tilde{s})\lvert w \rvert} \cdot \tilde{d}(w') \textbf{ or } \tilde{d}(w'w) \geq 2^{(1 - \tilde{s})(2^n + \ell)} \text{ for some } \ell \text{ such that } q \cdot 2^n \leq \ell \leq 2^n. 
\end{align*} 

The analysis to prove that for $N=2^n$,
 \begin{align*}
\pr_{(x,r) \in \Sigma^{s \cdot N} \times \Sigma^{s \cdot N}}[A(G_{N}(x)) = 1] \geq \frac{1}{N^2}.
\end{align*}
is quite similar to the one in the proof of Lemma \ref{lem:equivalenceforward}. The only difference is that instead of Lemma \ref{lem:BorelCantellionsgale}, we use Lemma \ref{lem:BorelCantellionsgale2}. 

Now, for any $n \in \N$, take any $r \in \Sigma^{s.2^n}$, and let $w' = g(r)$. 
Using Lemma \ref{lem:KolmogorvInequality}, for any $w'  \in \Sigma^{2^{n}}$, the number of strings $w \in \Sigma^{2^{n}}$ such that $d(w'w) \geq 2^{2^n.(1-\tilde{s})} d(w')$ is less than $2^{2^n} / 2^{2^n ( 1- \tilde{s})}$. Again, using Lemma \ref{lem:KolmogorvInequality}, the number of strings $w'w \in \Sigma^{2^{n+1}}$ such that $\tilde{d}(w'w) \geq 2^{(1 - \tilde{s})(2^n + \ell)}$  for some $\ell$ such that $q \cdot 2^n \leq \ell \leq 2^n$ is at most $2^{2^{n+1}}/2^{(1 - \tilde{s})(2^n + q 2^n)}$. Let $N = 2^n$. Now, using the union bound it follows that,

\begin{align*}
	\pr_{(x,r) \in \Sigma^N \times \Sigma^N} [A(x) = 1] < \frac{1}{2^{ N ( 1- \tilde{s})}} + \frac{1}{2^{(1 - \tilde{s})(1+q)N}}.
\end{align*}
Notice that $q$ is a constant which depends only on $\tilde{s}$ and $s''$.

It follows that if one-way functions exist, for every $s'<1$ with $2^{s'}\in \Q$ there exist $X \in g(\Sigma^\infty)$ such that $\cdimp({X})\geq s'$. Since $\cDimp (X) \geq \cdimp(X)$, for every $s' \in (2s,1)$ with $2^{s'}\in \Q$ there exist $X \in g(\mathcal{S})$ such that $\cDimp({X})\geq s'$. Now, Theorem \ref{thm:strongdimensionandgapcombined} follows from the observation that by taking $m$ large enough and choosing an appropriate $s'$, the quantity $s'-2s=s'-2^{-{m-1}}$ can be made arbitrarily close to $1$.

We obtain the following corollary of Theorem \ref{thm:strongdimensionandgapcombined}. These results provide a strong negative answer to the open question posed by Stull in \cite{stullsurvey} under the assumption that one-way functions exist.

 \begin{corollary}
If one-way functions exist then for every $\epsilon \in (0,1)$ there exists $\mathcal{F} \subseteq \Sigma^\infty$ such that,
\begin{align*}
\Kpoly(\F)\leq \Kstrpoly(\F) <\epsilon<1-\epsilon< \cdimp(\F) \leq \cDimp(\F)
\end{align*}
\end{corollary}
\begin{corollary}
If one-way functions exist then for every $\epsilon \in (0,1)$ there exists $ X \in \Sigma^\infty$ such that,
\begin{align*}
 \Kpoly(X) \leq \Kstrpoly(X) <\epsilon<1-\epsilon<  \cdimp(X) \leq \cDimp(X).
 \end{align*}
\end{corollary}

\section{Remarks on dimension gaps and infinitely-often one-way functions}
\label{sec:ioowfs}

Our main theorem (Theorem \ref{thm:equivalencetheorem}) shows that one-way functions implies the existence of polynomial time samplable distributions with respect to which dimension gaps exist on a  probability $1$ set. This assertion in turn implies the existence of infinitely-often one-way functions. In this section we remark on the possibility of proving that either one-way functions or infinitely-often one-way functions are equivalent to the existence of polynomial time samplable distributions with respect to which dimension gaps exist on a  probability $1$ set. 

Consider the statement of Lemma \ref{lem:equivalenceforward}. It is unlikely that the weaker hypothesis of existence of infinitely-often one-way functions yields the conclusion of Lemma \ref{lem:equivalenceforward} using the methods in the proof of the main theorem. In order to prove the conclusion of the lemma from the existence of infinitely-often one-way functions we need to break the PRG $\{G_n\}$ at all but finitely many lengths $n$. In our proof method, this requires that the following condition from Lemma \ref{lem:BorelCantellionsgale},
\begin{align*}
	\pr_{x \sim \Sigma^{s.2^n}}  [\tilde{d}(g(x )) > 2^{(1-\tilde{s})2^n}\tilde{d}(g(x  \restr s.2^{n-1}))] \geq \frac{1}{n^2}
\end{align*}
is true for all but finitely many $n$. This in turn requires that the following condition from Lemma \ref{lem:optimalmartingaleconstruction},
\begin{align*}
	\tilde{d}(X \restr 2^n) > 2^{(1-\tilde{s})2^{n-1}} \tilde{d}(X \restr 2^{n-1})
\end{align*}
is true for all but finitely many $n$ whenever there exist an $s'$-gale $d'$ that wins on $X$. However, the existence of an $s'$-gale $d'$ that wins on $X$ need not imply that there exists a martingale that roughly doubles it capital on half of the bits of $X$ from index $2^{n-1}$ to $2^n$ for all but finitely many $n$. If this happens, then by transforming the martingale $\tilde{d}$ to an $s'$-gale $d'$ (using the inverse of the transformation in Lemma \ref{lem:galetofastgrowingmgale}), we obtain that
\begin{align*}
	d'(X \restr 2^n) > \frac{1}{2^{(\tilde{s}-s')2^{n-1}}} \cdot d'(X \restr 2^{n-1})
\end{align*}
for all but finitely many $n$. However, an $s'$-gale $d'$ winning on $X$ need not respect this bound for all but finitely many $n$. $d'$ may violate this inequality infinitely many times and still manage to gain sufficient capital infinitely often to win on $X$. 

Now, consider the statement of Lemma \ref{lem:equivalenceconverse}. The main technical hurdle in obtaining a one-way function in the conclusion of Lemma \ref{lem:equivalenceconverse} is that assuming the non-existence of one-way functions, the approximation algorithm $\mathcal{A}$ from Theorem \ref{thm:irsextrapolation} may only succeed in approximating $\mathcal{D}_n$ for infinitely many values of $n$. Therefore, the $s'$-supergale $d$ in the proof of Lemma \ref{lem:equivalenceconverse} can only be approximated on infinitely many input string lengths $n$. 

However using the same techniques as in the proof of Lemma \ref{lem:equivalenceconverse}, we can obtain the following lemma which shows that if dimension gaps exists between $\cdimp$ defined using gales which are only approximable on infinitely many lengths and $\Kpoly$, then one-way functions exist. 

\begin{lemma}
\label{lem:equivalenceowfconverse}
If for some $s<1$ there exist a polynomial time samplable distribution $\nu$ over $\Sigma^\infty$ such that:
	\begin{enumerate}
		\item The number of random bits used by the sampler for $\nu$ on input $1^n$ is at most $sn$.
		\item There exist some $\hat{s} \in(s,1]$ such that for every $s' \in (s,\hat{s})$ and any polynomial-time $s'$-supergale $d$ that is \textbf{$\mathbf{\nu}$-approximable on infinitely many input lengths}, we have $\nu(S^\infty(d))=0$.
	\end{enumerate}
	Then, \textbf{one-way functions} exist.	
\end{lemma}

The proof is essentially identical to that of Lemma \ref{lem:equivalenceconverse} except that we use the universal extrapolation result Theorem 19 instead of Theorem 20 from \cite{ilango2022robustness}.

\section{Additional Proofs}
\label{sec:technicallemmas}

Below we give an alternative proof of the fact that $\cdimp$ dominates $\Kpoly$ from \cite{dercjournal,dercconference}. We also prove the analogous inequality for polynomial time strong dimension.

\begin{theorem}[\cite{dercconference,dercjournal}] \label{thm:hitchcockvinodchandran2}%
	For every $\F \subseteq \Sigma^\infty$, $\Kpoly (\F) \leq \cdimp(\F)$ and $\Kstrpoly (\F) \leq \cDimp(\F)$.
\end{theorem}

We need the following lemma to prove the above theorem.
\begin{lemma}
\label{lem:dyadicintervallemma}
Let $[a,b]$ be a sub-interval of $[0,1]$. Then, there exist $m \leq \lfloor \log (1/(b-a)) \rfloor$ and $j < 2^{m+2}$ such that,
\begin{align*}
	\left[\frac{j}{2^{m+2}},\frac{j+1}{2^{m+2}}\right] \subseteq [a,b].
\end{align*}
\end{lemma}
\begin{proof}
Let $m=\lfloor -\log(b-a) \rfloor$. Among the set of all intervals of the form $[k/2^m,k+1/2^m]$, let $I_m$ be the interval which maximizes the length of $I_m \cap [a,b]$ (if there are two such intervals, choose any of them). Consider the following sub-intervals of $I_m$,
\begin{align*}
	\left[\frac{4k}{2^{m+2}},\frac{4k+1}{2^{m+2}}\right], \left[\frac{4k+1}{2^{m+2}},\frac{4k+2}{2^{m+2}}\right], \left[\frac{4k+2}{2^{m+2}},\frac{4k+3}{2^{m+2}}\right] \text{ and }\left[\frac{4k+3}{2^{m+2}},\frac{4k+4}{2^{m+2}}\right].
\end{align*}
If none of the above are completely contained in $[a,b]$, then the length of $I_m \cap [a,b]$ is at most $1/2^{m+2}$. From the choice of $I_m$, it follows that then length of $[a,b]$ is at most $1/2^{m+1}$. Therefore,  $m=\lfloor -\log(b-a) \rfloor \geq m+1$ which is a contradiction. Therefore, one of the above sub-intervals of $I_m$ satisfies the required condition.
\end{proof}
Now, we prove Theorem \ref{thm:hitchcockvinodchandran}.
\begin{proof}[Proof of Theorem \ref{thm:hitchcockvinodchandran}]
	 Let $d: \Sigma^\infty \to \Q$ be an exact $s$-gale  that runs in time $t$ such that $d$ succeeds on all $X \in \F$. We show that for all $X \in \F$, $\liminf_{\nifty} \frac{K_{t'}(X \restr n)}{n} \leq s$, where $t'(n) = \poly(t(n))$. 
	
	Given $n \in \N$, let $S_n = \{x \in \Sigma^n : d(x) > 1\}$.  From the Lemma \ref{lem:KolmogorvInequality}, it follows that $|S_n| < 2^{sn}$. For $x \in S_n$, we now define an encoding $\mathcal{E}(x)$ that takes at most $sn + O(1)$ bits. 
	
	\textbf{Encoding}: For any $n \in \N$ and $x \in \Sigma^{n}$, letting $p(x) = d(x). 2^{-sn}$ we obtain a probability distribution over $\Sigma^n$. For every $n$, define the cumulative probability $c_n(x) = \sum_{y \in \Sigma^n ,\; y < x} p(y)$, where $y < x$ means that $y$ is less than $x$ in the lexicographic ordering.  If $d(x) > 1$, then $p(x) > 2^{-sn}$. Hence $c_n(x+1)-c_n(x)> 2^{-sn}$. Here $x+1$ is the $n$-length successor of $x$ in the lexicographic order. Also, we assume that $c_n(1^n+1)=1$.
	
	From Lemma \ref{lem:dyadicintervallemma}, there exist a dyadic interval $[j/2^{m+2},(j+1)/2^{m+2}]$ that is totally contained in $[c_n(x),c_n(x+1)]$. Define $\mathcal{E}(x)$ to be the binary encoding of  this dyadic interval. From Lemma \ref{lem:dyadicintervallemma}, we have 
	\begin{align*}
		\lvert \mathcal{E}(x) \rvert \leq m+2 \leq \left\lfloor \log \frac{1}{c_n(x+1)-c_n(x)} \right\rfloor +2 \leq  \lfloor \log 2^{sn} \rfloor +2 = sn + O(1).
	\end{align*}	
	Hence, it follows that $\mathcal{E}$ defines a one to one encoding from $S_n \to \Sigma^{*}$ such that $|\mathcal{E}(x)| < sn +O(1)$. for every $x \in S_n$.

	\textbf{Decoding}: We will now give a $\poly(t(n))$ procedure that uses $d$ to generate an $x \in \Sigma^n$, given $\mathcal{E}(x)$ as input. 
	The idea is that we perform a binary search over $\Sigma^n$ using the exact polynomial time computable $s$-gale $d$ and the input encoding $\mathcal{E}(x)$.
	
	Given $\mathcal{E}(x) \in \Sigma^{sn}$, observe that the dyadic interval $[j/2^{m+2},(j+1)/2^{m+2}]$ encoded by $\mathcal{E}(x)$ is totally contained in $[c_n(x),c_n(x+1)]$ for some $x \in \Sigma^n$.  
		
	\begin{algorithm} [H]
		\caption{Decoding Algorithm. Input $w$}
		\begin{algorithmic}[1]
			\State Set $x_0 = \lambda$, $c_0(x_0) = 0$
			\State Compute the dyadic interval $I$ such that $w$ is an encoding of $I$ 
			\For {$i = 0$ to $n-1$}
			\If {$I \subseteq [c_i(x_i),c_i(x_i) + p(x_i.0)]$}
			\State Set $x_{i+1} = x_i.0$
			\State Set $c_{i+1}(x_{i+1}) = c_{i+1}(x_i)$
			\Else
			\State Set $x_{i+1} = x_i.1$
			\State Set $c_{i+1}(x_{i+1}) = c_{i+1}(x_i) + p(x_i.0)$
			\EndIf
			\EndFor
		\end{algorithmic}
	\end{algorithm}
	
	\textbf{Running time of the algorithm:} Since $d$ is exact computable $2^{-s}$ is rational and hence $p(y)=d(y)2^{-s\lvert y \rvert}$ can be computed in time polynomial in $t(\lvert w \rvert)$ for any $y$ with length less than or equal to that of $w$. Every other step in the algorithm can be performed in polynomial time. Therefore on input $w$, the algorithm terminates in $\poly(t(n))$ steps.
	
	\textbf{Correctness:} At each step in the main loop of the algorithm, we keep track of the cumulative probability $c_{i}(x_i)$ for the current node as a running sum. On input $\mathcal{E}(x)$ for some $x \in \Sigma^n$, we are guaranteed that the interval $I$ encoded by $\mathcal{E}(x)$ satisfies $I \subseteq [c_n(x),c_n(x+1)]$. At step $0$, either $I \subseteq [0, p(0)]$ or $I \subseteq [p(0),p(1)]$. A simple inductive argument shows that at every step $i$, either $I \subseteq [c_i(x_i),c_i(x_i) + p(x_i.0)]$ or $I \subseteq [c_i(x_i)+ p(x_i.0),c_i(x_i)+ p(x_i.0) + p(x_i.1)]$. 
	
	The above conclusion implies that at the end of step $i$, $I \subseteq [c_{i+1}(y),c_{i+1}(y+1)]$ for some $y \in \Sigma^{i+1}$. therefore, at the end of the loop $I \subseteq [c_n(y),c_n(y+1)]$ for some $y \in \Sigma^n$. From the definition of the encoding $\mathcal{E}$, it follows that the algorithm on input $\mathcal{E}(x)$ for $x \in \Sigma^n$ terminates with $x_n=x$.

	\textbf{Conclusion:} Given $\F \subseteq \Sigma^\infty$ and $s > \cdimp(\F)$ such that $2^{s} \in \Q$.
	There exists an $s$-gale $d$ that succeeds on all $X \in \F$, we have that for every $X \in \F$, for infinitely many $n \in \N$, $d(X \restr n) \geq 1$. From the above algorithm it follows that, $\liminf_{n \to \infty} {K_{\poly(t(n))}} / {n}  \leq s$ and so $\Kpoly (\F) \leq s$. Hence we conclude that for every $\F \subseteq \Sigma^\infty$, $\Kpoly (\F) \leq \cdimp(\F)$. This proves the first inequality.
	
	Similarly, let $s > \cDimp(\F)$. then there exists an $s$-gale $d$ such that $d$ strongly succeeds on all $X \in \F$. It follows that for every $X \in \F$, for all but finitely many $n \in \N$, $d(X \restr n) \geq 1$.  Therefore using the above algorithm we obtain, $\limsup_{n \to \infty} {K_{\poly(t(n))}} / {n}  \leq s$ and so $\Kstrpoly(\F) \leq s$. Hence we conclude that for every $\F \subseteq \Sigma^\infty$, $\Kstrpoly(\F) \leq \cDimp(\F).$
\end{proof}

\section*{Acknowledgments}
The authors would like to thank CS Bhargav, Laurent Bienvenu,  John Hitchcock, Jaikumar Radhakrishnan,  Mahesh Sreekumar Rajasree and Marius Zimand for helpful comments and discussions.

\bibliographystyle{alphaurl}
\bibliography{fair001,main,random}

\end{document}